\documentclass[a4paper,USenglish]{lipics-v2021}

\usepackage{algorithm}
\usepackage[noend]{algorithmic}

\usepackage{graphicx}
\usepackage{xr}
\usepackage{amsmath,amsthm,amssymb}
\usepackage{placeins}
\usepackage{mathtools}
\usepackage{xcolor}

\usepackage{tikz}
\usetikzlibrary{arrows.meta}

\newcommand{\isim}[1]{\mathrel{\overset{#1}{\sim}}}

\newcommand{\IFRETURN}[2][]{%
  \STATE \textbf{if} #2 \textbf{then }%
  \ifx\relax#1\relax
    \textbf{return}%
  \else
    \textbf{return} #1%
  \fi
}

\newcommand{\PROCEDURE}[1]{%
  \STATE \textbf{Procedure} #1:%
  \begin{ALC@g}%
}

\newcommand{\PROCEDUREREQ}[2]{%
  \STATE \textbf{Procedure} #1: \COMMENT{Require: #2}%
  \begin{ALC@g}%
}

\newcommand{\ENDPROCEDURE}{\end{ALC@g} \item []}

\newcommand{\ignore}[1]{}

\newcommand{\execType}{component-serialized }

\newcommand{\ExecType}{Component-serialized }

\hideLIPIcs
\nolinenumbers

\title{Adaptive Snapshots Require Visible Reads}

\author{Niv Sulimany}{Technion – Israel Institute of Technology, Haifa, Israel}{sulimany.niv@campus.technion.ac.il}{https://orcid.org/0009-0002-2900-9204}{}

\author{Tomer Cory}{Technion – Israel Institute of Technology, Haifa, Israel}{tomer.cory@campus.technion.ac.il}{https://orcid.org/0009-0008-6977-0341}{}

\author{Erez Petrank}{Technion – Israel Institute of Technology, Haifa, Israel}{erez@cs.technion.ac.il}{https://orcid.org/0000-0002-6353-956X}{}

\authorrunning{N. Sulimany, T. Cory, and E. Petrank} 

\Copyright{Niv Sulimany, Tomer Cory, and Erez Petrank} 



\funding{This research was supported by THE ISRAEL SCIENCE FOUNDATION grant No. 1102/21.}

\EventEditors{Ioannis Chatzigiannakis, Andrea Vitaletti, Keren Censor-Hillel, and William K. Moses Jr.}
\EventNoEds{4}
\EventLongTitle{40th International Symposium on Distributed Computing (DISC 2026)}
\EventShortTitle{DISC 2026}
\EventAcronym{DISC}
\EventYear{2026}
\EventDate{November 9--13, 2026}
\EventLocation{Rome, Italy}
\EventLogo{}
\SeriesVolume{397}
\ArticleNo{1}

\begin{document}

\maketitle

\begin{abstract}

Snapshots are widely used to record the state of a running execution. Snapshots have been extensively studied in the literature, with the goal of improving performance and extending functionality. In this work, we consider {\em adaptive} snapshots over a set of $m$ components. Adaptive snapshots provide a \textsc{Click}() operation that logically creates a new snapshot and an \textsc{Observe}$(i)$ operation that returns the state of component~$i$ at the most recent \textsc{Click}. Several constructions of adaptive snapshots have recently been proposed; interestingly, none of them employs invisible reads, even though invisible reads can improve performance, sometimes significantly. In this paper, we ask whether it is possible to build an adaptive snapshot with invisible reads. We show that, even when restricting the snapshot algorithm to the single-writer, single-scanner setting, under reasonable assumptions satisfied by all existing adaptive snapshot implementations, adaptive snapshots with invisible reads are not linearizable.

\end{abstract}

\section{Introduction} \label{sec:introduction}


An atomic \textsl{snapshot}
is a fundamental concept in concurrent and distributed computing that has been extensively studied and widely used. Atomic snapshots were formalized by Afek et al.~\cite{atomic_snapshot_attiya_def} and Anderson~\cite{composite_registers}. In essence, the problem consists of $n$ concurrent processes maintaining an array of $m$ components while supporting two operations: \textsc{Update}$(i,x)$, which sets the value of component~$i$ to $x$, and \textsc{Scan}(), which returns a vector containing the current values of all components.

Early approaches~\cite{atomic_snapshot_attiya_def, composite_registers, 10.1145/3382734.3406005, lattic_mwmr, snapshot_lattice, atomic_snapshot_nlogn, 10.1145/1281100.1281108, 10.1145/1146381.1146416} focused on implementations using only read-write registers, typically under a single-writer assumption. 
Subsequent work~\cite{f_array} relaxed these restrictions by allowing stronger synchronization primitives, such as Compare-and-Swap and Fetch-and-Increment, thereby extending functionality and improving performance. Jayanti~\cite{jayanti_snapshot}, extending Riany, Shavit, and Touitou~\cite{strong_single_writer}, presented a multi-writer snapshot algorithm with $O(1)$ step complexity for \textsc{Update} operations, $O(m)$ step complexity for \textsc{Scan} operations, and space complexity $O(mn^2)$, later reduced by Ba~\cite{Ba2006Snapshot} to $O(mn)$. Although efficiency improved, the resulting performance overhead was still high for many settings.

\textsl{Adaptive snapshot objects}~\cite{adaptive_partial_snapshot, constant_time_snapshot_cas} remove the requirement that all snapshot locations be known in advance by supporting three types of operations: \textsc{Update}$(i, val)$, which updates the value of component~$i$ and potentially returns its state, \textsc{Click()}, which  creates a new snapshot, and \textsc{Observe}$(i)$, which returns the value of component~$i$ at the last \textsc{Click} operation. An adaptive snapshot object, unlike the standard snapshot object, allows the user to observe the memory items adaptively, based on previous items already observed.

Wei et al.~\cite{constant_time_snapshot_cas} devised an adaptive snapshot algorithm for concurrent data structures built from CAS objects, 
with the progress guarantees of the original data structures preserved, and with step complexity $O(1)$ for the \textsc{Update}, \textsc{Click}, and \textsc{Observe} operations. 
The step complexity for reading a version of a CAS object while traversing the data structure is proportional to the number of successful CAS operations applied to the item after the snapshot was taken and before the read.


Bashari and Woelfel~\cite{adaptive_partial_snapshot} presented a single-writer adaptive snapshot for an array of $m$ locations. Its \textsc{Click} operation has $O(1)$ step complexity, whereas its \textsc{Observe} and \textsc{Update} operations have $O(\log n)$ step complexity. Bashari et al.~\cite{bashari_et_al:LIPIcs.DISC.2024.7} extended this implementation to support multiple scanners and updaters, with constant step complexity for \textsc{Click} and \textsc{Update} operations and $O(\log n)$ step complexity for \textsc{Observe} operations

Jayanti, Jayanti, and Jayanti~\cite{memsnap_podc_24} presented a single-scanner adaptive snapshot object over an array of components. Its \textsc{Update} operations may perform any RMW operation supported by the underlying hardware, and every operation has constant time complexity. This was subsequently~\cite{shared_archive} extended to support an archive of snapshots supporting multiple scanners, with constant-time \textsc{Update} and \textsc{Click} operations, and with a \textsc{Search}$(c,i)$ operation that returns the state of component $c$ in the $i$-th snapshot and has step complexity linear in the number of snapshots taken since the requested snapshot.

Adaptive snapshots seem to be a promising target for further performance improvements and a candidate for use in practice. Moreover, important garbage collection algorithms use variants of adaptive snapshots to build concurrent garbage collectors~\cite{Steele75,DLG,snapshot_collector}. However, all known adaptive snapshot algorithms over $m$ components require each update operation on a component to cooperate with scans and potentially report modifications. Moreover, even reads of individual components must cooperate with scans and may need to report the observed state to the scanner. While the cost of such cooperation is typically $O(1)$, read operations are usually far more frequent than writes, and substantial effort is often devoted to making read operations efficient and wait-free~\cite{TheArt}. 
Reporting read operations to the garbage collector is considered undesirable and detrimental to overall performance~\cite{blackburn2004barriers,atikoglu2012workload}, and is therefore typically avoided~\cite{garbage_collection_handbook,Steele75,DLG,snapshot_collector}. In the exceptional cases where overhead is added to read operations in order to allow concurrent defragmentation of the heap~\cite{pizlo2008study,flood2016shenandoah,10.1145/800055.802042}, a noticeable reduction in performance has been reported.  
Similarly, in adaptive snapshots, such additional overhead can degrade performance and impose significant costs on operations that access individual components, potentially limiting applicability. A natural question is whether the overhead added to read operations can be eliminated or at least made invisible. An {\em invisible read} is one that does not write to shared memory. Invisible reads reduce cache-coherence contention, since read operations need not access cache lines containing shared variables in \textit{exclusive} mode~\cite{hennessy2011computer}.

In this paper, we show that such an optimization is not possible, even in the simplest single-writer, single-scanner setting. In particular, we show that under natural assumptions satisfied by all current implementations of adaptive snapshots~\cite{constant_time_snapshot_cas, shared_archive, memsnap_podc_24, bashari_et_al:LIPIcs.DISC.2024.7, adaptive_partial_snapshot}, linearizable adaptive snapshots cannot employ invisible reads of the individual components. The assumptions we make are formalized in Section~\ref{sec:problem}. Informally, they include \emph{obstruction-freedom} (the weakest standard progress guarantee); \emph{oblivious click}, which requires \textsc{Click} operations not to access the components; \emph{update linearization independence from \textsc{Click}}, which requires that \textsc{Click} operations do not interfere with the linearization of pending updates to components; and \emph{component encapsulation}, which requires that components are accessed only through their methods and that the point at which a component update is linearized consists solely of a write to that component.

The assumptions above are natural in the design of an adaptive snapshot object. Progress guarantees are clearly desirable. Oblivious \textsc{Click}, or at least $k$-bounded \textsc{Click} operations, are important for keeping \textsc{Click} efficient; indeed, in all existing work, the \textsc{Click} operation has complexity $O(1)$. Component encapsulation is a standard and desirable software-engineering property that simplifies both the design and implementation of the object. Finally, the independence of \textsc{Click} from updates may be the least obvious of these assumptions, but it is desirable because it simplifies the \textsc{Click} operation. Moreover, this property, like all of the other assumptions, is satisfied by all existing designs.

Nevertheless, we do not claim that these assumptions are unavoidable or impossible to relax. Rather, the results in this paper can be interpreted as identifying the precise properties that make such implementations impossible. By relaxing some of these properties, it may be possible to design efficient algorithms with invisible reads and eventually deploy them in practical systems. In essence, our results show that, with invisible \textsc{Read} operations, processes executing \textsc{Click} operations and processes updating components must communicate with one another.

We call an algorithm that satisfies the properties discussed above an \emph{adaptive snapshot with oblivious click}. In Section~\ref{sec:first_impossible}, we prove 
that such algorithms cannot employ \emph{invisible} \textsc{Read} operations. The proof constructs two executions that are indistinguishable to a process $p$ executing a \textsc{Click} operation, yet in one execution a concurrent update must be linearized before the \textsc{Click}, whereas in the other it must be linearized after the \textsc{Click}. See Section~\ref{sec:preliminaries} for the definitions of linearizability and indistinguishability. A subsequent \textsc{Observe} operation executed by $p$ must therefore return an incorrect result in one of the executions, contradicting linearizability.

In Subsection~\ref{subsec:non-intrusive}, we relax the restriction that \textsc{Click} operations do not access components. We show that the same impossibility result holds even when a \textsc{Click} operation may access a bounded number of components, under a slightly stronger update-linearization-independence property. We call such algorithms \emph{adaptive snapshot with a $k$-bounded click}. In Section~\ref{sec:second_impossible}, we show that these algorithms also cannot employ invisible \textsc{Read} operations.

\paragraph*{Organization.}
We begin in Section~\ref{sec:related_work} with a survey of previous work most relevant to our results. We proceed in Section~\ref{sec:preliminaries} by describing the computational
model, defining the adaptive snapshot object, and reviewing the notion of
\emph{linearizability}. In Section~\ref{sec:problem}, we formally define the two
sets of assumptions that we make on adaptive snapshot algorithms. We then state the main
theorems of the paper. Section~\ref{sec:first_impossible} presents the proof of
Theorem~\ref{th:main_oblivious}, showing the impossibility of invisible reads for \emph{adaptive snapshots with oblivious click}. Section~\ref{sec:second_impossible}
presents the proof of Theorem~\ref{th:main_non_intrusive}, asserting the impossibility of invisible reads for \emph{adaptive snapshots with a $k$-bounded click}. We conclude in
Section~\ref{sec:conclusion}. Appendix~\ref{sec:additional_proofs} contains some proof details omitted from Section~\ref{sec:first_impossible} and Section~\ref{sec:second_impossible} due to space constraints.

\section{Related work} \label{sec:related_work}

In addition to the works highlighted in the introduction, several lines of research have explored the design, complexity, and limitations of snapshot objects. We briefly survey the results most relevant to our work.

One line of work, aiming at improved performance, studied snapshot constructions over concurrent data structures~\cite{constant_time_snapshot_cas,petrank_snapshot,
7161513,10.1145/3007748.3007771}, rather than over arrays of  components. A snapshot-based technique for concurrent data structures that implement a set or a dictionary was presented by Petrank and Timnat~\cite{petrank_snapshot}, achieving $O(1)$ time complexity for \textsc{Update} and \textsc{Read} operations, whereas taking a snapshot requires invoking \textsc{takeSnapshot}, which returns an iterator over a snapshot of the data structure and takes $O(m)$ time. The benefit of snapshots over data structures is that the overhead on update operations is reduced, because a report is required once per data structure operation rather than for each memory access. This technique was later extended to support range queries by Chatterjee~\cite{10.1145/3007748.3007771}.

In another line of work aimed at improving performance, Attiya, Guerraoui, and Ruppert~\cite{partial_snapshot} formalized the notion of {\em partial snapshots}. Instead of scanning every component during a \textsc{Scan}, it is possible to scan only a subset of the locations that are needed for the computation. 
In scenarios where the number of needed locations is much smaller than the total number of locations $m$, the step complexity required to collect a partial snapshot can be substantially lower than that of a full snapshot. However, it requires knowing, at the time of the scan, the entire set of needed locations. 
Range queries have been studied extensively since then~\cite{10.1145/3503221.3508412, sheffiOPODIS22, 10.1145/3323165.3323197, 10.1145/2692916.2555267, 10.1145/3200691.3178489}.

The study of snapshot objects has produced a long sequence of lower bounds and impossibility results that reveal the inherent costs of maintaining a consistent view of shared memory in asynchronous systems. These results have characterized limitations on complexity measures such as time, space, and synchronization requirements.

Israeli and Shirazi~\cite{ISRAELI199833} showed that, for a single-writer snapshot implementation over $n$ processes using single-writer registers, the worst-case step complexity of an \textsc{Update} operation is $\Omega(n)$, even in executions without concurrent updates.

Fatourou et al.~\cite{10.1145/780542.780582,10.1145/1314690.1314694} employed covering arguments to show that any snapshot implementation over~$m$ components requires at least~$m$ multi-writer registers, and that in a system with~$n$ processes, the worst-case step complexity of a \textsc{Scan} operation is $\Omega(nm)$. Attiya et al.~\cite{ATTIYA20111570} extended these techniques to \emph{partitioned} implementations, namely, implementations that may use an arbitrary number of base objects, provided that each base object is modified only by \textsc{Update} operations on one specific component.

A recent result by Castañeda and Hernández Martínez~\cite{rmw_invisble_read_full_snapshot} presents a wait-free snapshot algorithm using only read and write operations. The algorithm does not provide an implementation of \textsc{Read} operations and is not adaptive in its original form. However, it can be made adaptive by implementing \textsc{Click} using the scan operation, having \textsc{Observe} read from the snapshot produced by the most recent scan, and implementing \textsc{Read} by simply reading and returning the current value of the corresponding component. This transformation would not yield an efficient adaptive snapshot algorithm, but it would produce one with invisible \textsc{Read} operations. Thus, it illustrates how avoiding one of our assumptions, in this case, oblivious click, can make an adaptive snapshot with invisible reads possible.

Finally, Attiya et al.~\cite{laws_of_order} proved that certain synchronization patterns are unavoidable in implementations of objects such as mutual exclusion objects, sets, queues, and stacks. In a similar spirit, our work shows that, under our assumptions, \textsc{Read} operations must modify shared memory.

 \section{Preliminaries} \label{sec:preliminaries}


\ignore{

\begin{definition}

\textsl{Strong linearizability}~\cite{attiya2019puttingstronglinearizabilitycontext, golab2011linearizableimplementationssufficerandomized, ovens2019stronglylinearizableimplementationssnapshots} strengthens
linearizability by requiring that the assignment of linearization points be
\emph{prefix-preserving}: for every execution $E$ and every prefix $E'$ of $E$,
the linearization points assigned to operations in $E'$ are identical to those
assigned to the same operations in $E$. Consequently, the linearization point
of an operation cannot depend on events that occur after its response.
\end{definition}

}

We consider the standard asynchronous shared memory model~\cite{distributed_computing,impossibility_results, 10.5555/1374804}. The system consists of $n$ processes, labeled $p_1,\ldots,p_n$. Each process has a local state that is accessible only to that process. Processes communicate by performing read and write operations on shared memory, as well as through \emph{read-modify-write} (RMW) instructions on shared memory. Available RMW instructions include widely supported operations such as Test-and-Set, Fetch-and-Add, and Compare-and-Swap (CAS) on a single memory word, as well as wide CAS (double-width CAS) on two adjacent memory words. However, we exclude operations such as $m$-CAS that can operate on multiple non-adjacent words and are not widely available in hardware. Such operations need to be implemented in software using available primitives~\cite{guerraoui_et_al:LIPIcs.DISC.2020.4, 10.5555/645959.676137, e105-d_5_946, 10.5555/645954.675655, 10.1145/197917.198079,TimnatHP15}.

A \emph{configuration} describes the state of the system at a given point in
time. It consists of the local states of all processes together with the values
of all shared-memory variables. An \emph{initial configuration} is a
configuration in which the local state of every process and the value of every
shared-memory variable equal their respective initial states.

We follow the definition given in~\cite{impossibility_results}: two
configurations $C$ and $C'$ are \emph{indistinguishable} to a process $p$ if
the local state of $p$ is identical in both configurations. We denote this by
$C \isim{p} C'$. If $P$ is a set of processes and
$C \isim{p} C'$ holds for every $p \in P$, we write $C \isim{P} C'$.


At any configuration, a set of execution steps, also referred to as events, may occur, each executed by a single process. For example, a process may read from or write to a shared variable, or invoke an RMW operation on a shared object. In each step, a process uses its local state to select a shared-memory object and an operation, applies the operation to the object, and updates its local state according to its transition system~\cite{distributed_computing}.

The set of events that can occur at a configuration depends on both the state of the shared memory and the local states of the processes. Formally, for a configuration $C$, the set of events that can occur at $C$ consists of all steps that may be taken by some process when the system is in configuration $C$.

An \emph{execution} is a sequence of alternating configurations and events, starting with the initial configuration, such that each event can occur in the configuration that precedes it, and its execution results in the configuration that follows it. For a set of processes $P$, an execution or sequence of events is said to be $P$-only if every event in it is performed by a process in $P$. The history associated with an execution is its sequence of events in the order in which they occurred in the execution.

A sequence of events $\sigma$ can occur starting at a configuration $C$ if there exists an execution that starts from $C$ whose sequence of events is exactly $\sigma$. If $\sigma$ is finite, we denote by $C\sigma$ the final configuration reached by executing $\sigma$ from~$C$.

In this model, the following lemma (Lemma 2.1 in~\cite{impossibility_results}) holds. It states that if the local states of a set of processes are initially identical and these processes read the same values during an execution, then their final local states are also identical.

\begin{lemma}\label{lemma:indistinguishable_configurations}
    Let $\sigma$ be a sequence of events performed by a set of processes $P$ that can occur starting at a configuration~$C$. If $C \isim{P} C'$ and all shared-memory locations accessed during $\sigma$ have the same values in $C$ and $C'$, then $\sigma$ can also occur starting at $C'$. Moreover, if $\sigma$ is finite, then $C\sigma \isim{P} C'\sigma$.
\end{lemma}

\begin{definition} [Obstruction-freedom]
An algorithm is \emph{obstruction-free} if any process~$p$ completes the operation it is performing whenever it executes in isolation for sufficiently many steps.
    
\end{definition}

\begin{definition}[Linearizability] \label{def:linearizability}

Linearizability~\cite{maurice_wing_linearize} is a standard correctness
criterion for concurrent objects. An execution $E$ consists of a set of
operations, each represented by an invocation event and a matching response
event. An execution is \emph{sequential} if each invocation is immediately
followed by its matching response (except possibly for a final unmatched
invocation).

An execution $E$ is \emph{linearizable} if there exists a mapping that assigns
to each operation a single event in the execution, called its \emph{linearization point},
such that:
\begin{itemize}
    \item \emph{Real-time order:}
    The linearization point of each operation occurs at some instant between
    its invocation and its response in $E$.
    \item \emph{Sequential correctness:}
    Ordering the operations according to the order of their linearization points in $E$ yields a legal sequential execution that is consistent with the
    object's sequential specification and in which each operation returns
    the same value as in $E$.
\end{itemize}

An algorithm is linearizable if every execution of the algorithm is linearizable.

\end{definition}

\begin{definition} [Adaptive snapshot object] \label{def:adaptive_snapshot}

An adaptive snapshot object consists of $m$ linearizable components. Each component must be readable, that is, it must support a linearizable operation that returns its state without modifying it.
The abstract state of the snapshot object is a pair of arrays of length~$m$,
denoted $(state, snap)$.
Intuitively, $state[j]$ stores the current state of the $j$-th component,
while $snap[j]$ stores the state of the $j$-th component in the most recent
snapshot.
Initially, $state = snap = \vec{s}$ for a given initial array of states
$\vec{s}$.

The adaptive snapshot object supports four operations: \textsc{Invoke}, \textsc{Click}, \textsc{Observe}, and \textsc{Read}. We specify the sequential behavior of each operation by describing the resulting state $(state', snap')$ and the return value~$r$ when the operation is invoked in state $(state, snap)$.

\begin{enumerate}
    \item \textsc{Invoke}$(i, op, arg)$: Applies the operation $op(arg)$ to the
    $i$-th component. $op$ must be one of the operations supported by the components. Specifically,
    (i) the snapshot remains unchanged, $snap' = snap$,
    (ii) the state of component~$i$ is updated to $state'[i] = state[i].op(arg)$,
    (iii) all other components remain unchanged,
    $\forall j \neq i:\; state'[j] = state[j]$, and
    (iv) the return value $r$ is the result of executing $state[i].op(arg)$.

    \item \textsc{Click}$()$: Takes a snapshot of the current state of the components. That is, $snap' = state' = state$, and the operation returns $r = \textsl{ack}$.

    \item \textsc{Observe}$(i)$: Returns the state of the $i$-th component
    in the most recent snapshot. Formally, $snap' = snap$, $state' = state$, and
    $r = snap[i]$.

    \item \textsc{Read}$(i)$: Returns the current state of the $i$-th component without modifying it. Formally, $state' = state$, $snap' = snap$, and $r = state[i]$.
\end{enumerate}

\end{definition}

\ignore{

\begin{definition} [Obstruction-free adaptive snapshot algorithm] \label{def:obstruction_free_snapshot}

Let $\mathcal{A}$ be an algorithm that implements the adaptive snapshot object (see definition~\ref{def:adaptive_snapshot}). If all the snapshot operations are obstruction-free, and all the operations supported by the components are obstruction-free, we call $\mathcal{A}$ an obstruction-free adaptive snapshot algorithm.
    
\end{definition}

}

\begin{definition}[Accessing a component] \label{def:component_access}
For an implementation of an adaptive snapshot object, we partition the shared-memory locations into locations associated with exactly one component and locations that are not associated with any particular component. A location associated with component~$i$ is used by the implementation to represent the current state of that component. We say that a process
\emph{accesses component}~$i$ when it accesses a
shared-memory location associated with component~$i$.
\end{definition}

We now formally define the notions of an invisible \textsc{Read} operation and an updating operation.

\begin{definition}[Invisible \textsc{Read}] \label{def:invisible_read}
A \textsc{Read} operation is \emph{invisible} if it performs no write or
read-modify-write operation on shared memory. 
\end{definition}
Invisible \textsc{Read} operations leave the
shared-memory state unchanged and and therefore cannot be detected by other processes through shared memory~\cite{10.1007/11864219_14, 10.1145/1345206.1345233, 10.1145/872035.872048, 10.1007/11561927_26, 10.1145/1073814.1073861}.

In this paper, we focus on a restricted class of executions of a snapshot object in which there is at most one concurrent \textsc{Invoke} operation per component. Such executions are of particular interest because the executions constructed in Sections~\ref{sec:first_impossible} and~\ref{sec:second_impossible} to establish the impossibility satisfy this property.

\begin{definition}[\ExecType execution] \label{def:exec_type}
Let $\mathcal{A}$ be an adaptive snapshot algorithm.
An execution $E$ of $\mathcal{A}$ is a` \execType execution
if, for every component~$i$, no two \textsc{Invoke}
operations on component~$i$ are concurrent in~$E$.
\end{definition}

\begin{definition}[Updating operation]\label{def:updating_operation}

Let $E$ be a linearizable execution of an adaptive snapshot algorithm $\mathcal{A}$, and let $OP=\textsc{Invoke}(i,op,arg)$ be an operation in $E$. We say that $OP$ is an \emph{updating operation} with respect to a linearization $L$ of $E$ if, when the operations in $E$ are executed sequentially according to the linearization order induced by $L$, $OP$ changes the state of component~$i$. If $OP$ is updating with respect to every linearization of $E$, then we simply say that $OP$ is an updating operation in $E$.

\end{definition}

We first observe that if $E$ is a linearizable \execType execution, then an operation is either updating with respect to every linearization of $E$ or with respect to no linearization, and the resulting state is the same in all linearizations. Consequently, such operations may simply be referred to as updating operations, without specifying a particular linearization. The formal statement appears in Observation~\ref{obs:updating_op_globabl} below.

\begin{observation}\label{obs:updating_op_globabl}
Let $E$ be a linearizable \execType\ execution of an adaptive
snapshot algorithm $\mathcal{A}$, and let
$OP=\textsc{Invoke}(i,op,arg)$ be an operation in~$E$.
Suppose that $OP$ is an updating operation with respect to some
linearization of~$E$, and that, when the operations are executed
sequentially according to this linearization, $OP$ changes the
state of component~$i$ from \texttt{state$_1$} to
\texttt{state$_2$}. Then, in every linearization of~$E$,
$OP$ is an updating operation, and when the operations are executed
sequentially according to the linearization, $OP$ changes the state
of component~$i$ from \texttt{state$_1$} to
\texttt{state$_2$}.
\end{observation}

Observation~\ref{obs:updating_op_globabl} is proved formally in Appendix~\ref{proof:updating_op_globabl}.

Next, we define the notion of an \emph{effective linearization
step}. This definition applies only to \execType\ executions,
namely, executions in which there are no concurrent
\textsc{Invoke} operations on the same component. Informally, the
effective linearization step of a completed updating operation is the step at
which the modification to the component takes effect and becomes
visible to subsequent \textsc{Read} operations. The effective
linearization step is not necessarily a step executed by the process
that invoked the updating operation. In executions with concurrent
\textsc{Invoke} operations on the same component (i.e.,
executions that are not component-serialized), effective linearization steps
are undefined.

\begin{definition}[Effective linearization step]
\label{def:effective_linearization}

Let $\mathcal{A}$ be a linearizable obstruction-free adaptive snapshot algorithm. Let $E$ be a \execType execution of $\mathcal{A}$, and let $\textsc{Invoke}(i,op,arg)$ be a completed updating operation (see Definition~\ref{def:updating_operation}) in $E$ that changes the state of the $i$-th component from \texttt{state$_1$} to \texttt{state$_2$}. Then the effective linearization step of the operation is the first step during its execution after which replacing the remaining suffix of $E$ by an isolated \textsc{Read}$(i)$ operation causes that operation to return \texttt{state$_2$}.
\end{definition}

In Observation~\ref{obs:elp_exist} below, we prove that every completed updating operation has an effective linearization step.

\begin{observation} \label{obs:elp_exist}

Let $\mathcal{A}$ be a linearizable adaptive snapshot algorithm. Assume that \textsc{Read} operations in $\mathcal{A}$ are obstruction-free. Let $E$ be a \execType execution of $\mathcal{A}$. Then, every completed updating operation in $E$ has a unique effective linearization step. 
\end{observation}

Observation~\ref{obs:elp_exist} is proved formally in Appendix~\ref{proof:elp_exist}.

\section{Problem statement} \label{sec:problem}



We aim to show that implementing invisible \textsc{Read} operations in adaptive snapshot algorithms is impossible under a natural set of assumptions. In this section, we identify two such sets of properties that characterize a broad class of adaptive snapshot algorithms~\cite{memsnap_podc_24,bashari_et_al:LIPIcs.DISC.2024.7,adaptive_partial_snapshot,shared_archive,constant_time_snapshot_cas}. The impossibility result we establish holds even for the simpler single-scanner setting, and therefore applies to multi-scanner snapshot implementations as well.

In this section (and throughout the paper), we denote by $m$ the number of components managed by the snapshot object, and by $n$ the number of processes in the system.


\subsection{Adaptive snapshot with oblivious click} \label{subsec:click_oblivious}

We begin with a natural class of algorithms in which \textsc{Click} operations are oblivious to the states of the underlying components, and the responsibility for recording and reporting component values is delegated to \textsc{Observe} and \textsc{Invoke} operations. We refer to such
algorithms as \emph{adaptive snapshots with oblivious click}, which we  formally define below.

We additionally assume that if an effective linearization step (see Definition~\ref{def:effective_linearization}) of an \textsc{Invoke} operation exists, then it is independent of concurrent \textsc{Click} executions. In particular, suppose the next execution step of process $p$ is the effective linearization step of an updating \textsc{Invoke} operation on some component. If a \textsc{Click} operation is executed just before $p$ performs that step, then the next step executed by $p$ still constitutes the effective linearization step of that \textsc{Invoke} operation.

\begin{definition}[Update linearization independence from \textsc{Click}]
\label{def:invoke_independence}
Let $\mathcal{A}$ be an adaptive snapshot algorithm and $E$ be a \execType execution of $\mathcal{A}$ with history
$e_1, e_2, \ldots$ (finite or infinite). Let $p$ be a process that executes an effective linearization step of an updating operation \textsc{Invoke}$(i,op,arg)$. Let $e_t$ denote the step that is the effective linearization step of this operation in $E$. Consider an execution $E'$ whose first $t-1$ events are identical to those of $E$, followed by an arbitrary sequence of \textsc{Click} steps executed by processes other than $p$ (either continuing
previously invoked operations or newly invoked ones). After these steps, process $p$ performs step $e_t$, followed by any valid suffix of execution
steps. In the resulting execution $E'$, the effective linearization step of
\textsc{Invoke}$(i,op,arg)$ is the event $e_t$, performed by $p$, and the
result of the operation (both the state modification and the returned value)
is identical to its result in execution $E$. 

We say that an adaptive snapshot algorithm $\mathcal{A}$ satisfies
\emph{Update linearization independence from \textsc{Click}} if the property defined above holds for every \execType execution $E$ of $\mathcal{A}$, every process $p$, and every
\textsc{Invoke}$(i,op,arg)$ operation in $E$.
\end{definition}

Following this definition, and using the fact that invisible \textsc{Read} operations do not modify shared memory, we immediately obtain that if concurrent \textsc{Click} operations do not affect the linearization of updating operations, then neither do invisible \textsc{Read} operations. This is formalized in the following lemma.

\begin{lemma}[Update linearization independence from \textsc{Click} and \textsc{Read}]
\label{lemma:invoke_independence_to_read}
Let $\mathcal{A}$ be an adaptive snapshot algorithm that satisfies update linearization independence from \textsc{Click}, and implements invisible \textsc{Read} operations. Let $E$ be a \execType execution of $\mathcal{A}$ with history
$e_1, e_2, \ldots$ (finite or infinite). Let $p$ be a process that executes an effective linearization step of an updating operation \textsc{Invoke}$(i,op,arg)$.
 Let $e_t$ denote the step that is the effective linearization step of this operation in $E$. Consider an execution $E'$ whose first $t-1$ events are identical to the first
$t-1$ events of $E$, followed by an arbitrary sequence of steps, not executed by $p$, consisting solely of \textsc{Click} and \textsc{Read} operations (either continuing
previously invoked operations or newly invoked ones), such that no process executes steps from both \textsc{Click} and \textsc{Read} operations in that sequence. After these steps, process $p$ performs step $e_t$, followed by any valid suffix of execution
steps. Then, in execution $E'$, the effective linearization step of
\textsc{Invoke}$(i,op,arg)$ is the event $e_t$, performed by $p$, and the
result of the operation (both the state modification and the returned value)
is identical to its result in execution $E$.
\end{lemma}

\begin{proof}
Intuitively, invisible reads are invisible to the updating process and do not affect their execution. The formal proof is postponed, due to lack of space, and appears in Appendix~\ref{proof:invoke_independence_to_read}.
\end{proof}

We are now ready to introduce the first class of algorithms.

\begin{definition}[Adaptive snapshot with oblivious click]\label{def:properties_oblivious_click}
A linearizable adaptive snapshot algorithm $\mathcal{A}$, which maintains a
collection $\mathcal{O}$ of $m$ linearizable components, is called an
\textsl{adaptive snapshot with oblivious click} algorithm if it implements the
adaptive snapshot object and satisfies the following properties:
\begin{enumerate}

\item \label{prop:oblivious_limit_click}
    \emph{Click-obliviousness:}
    A \textsc{Click} operation does not access any component (see Definition~\ref{def:component_access}).

\item \label{prop:oblivious_obstruction_free}
    \emph{Progress guarantee:}
    $\mathcal{A}$ is an obstruction-free adaptive snapshot algorithm.

\item \label{prop:oblvious_click_linearize}
    \emph{Update linearization independence from click:}
    $\mathcal{A}$ satisfies \emph{Update linearization independence from \textsc{Click}} as in Definition~\ref{def:invoke_independence}.
    
    \item \label{prop:oblivious_act_on_memory}
    \emph{Component encapsulation:}
    Component~$\mathcal{O}[i]$ is modified only through \textsc{Invoke}$(i,op,arg)$ operations. Moreover, in \execType executions, the effective linearization step of
    an updating operation, accesses and modifies only the portion of shared memory associated with component~$\mathcal{O}[i]$.
\end{enumerate}
\end{definition}


The requirements in Definition~\ref{def:properties_oblivious_click} are natural, and all non-blocking adaptive snapshot algorithms that we are aware of~\cite{bashari_et_al:LIPIcs.DISC.2024.7,shared_archive,memsnap_podc_24,constant_time_snapshot_cas,adaptive_partial_snapshot} satisfy these assumptions. As an illustration, in Appendix~\ref{sec:examples_satisfy_props} we explicitly show that the algorithms presented in~\cite{memsnap_podc_24} and~\cite{bashari_et_al:LIPIcs.DISC.2024.7} satisfy the required properties. Of-course, it is possible to construct contrived implementations that do not satisfy Definition~\ref{def:properties_oblivious_click}.

Property~\ref{prop:oblivious_limit_click} requires that \textsc{Click} operations do not directly access the components in order to record their state. As a result, components' states can only be recorded lazily, for example, during an
\textsc{Observe}, \textsc{Invoke}, or \textsc{Read} operations.

Property~\ref{prop:oblivious_obstruction_free} requires the algorithm to satisfy
a minimal progress guarantee, specifically, obstruction-freedom. This excludes implementations
that block all operations in order to initiate a snapshot.


Property~\ref{prop:oblvious_click_linearize} states that \textsc{Click} operations do not interfere with a pending effective linearization step of an updating operation. Note that a \textsc{Click} operation may still affect the actual linearization point of an \textsc{Invoke} operation. However, when an effective linearization step is about to be executed by process $p$, an execution step of a \textsc{Click} operation by another process does not prevent the next step of $p$ from being an effective linearization step, as observed by a solo \textsc{Read} operation.

Finally, component encapsulation (Property~\ref{prop:oblivious_act_on_memory}) prohibits simultaneously modifying multiple objects, as well as simultaneously modifying a component and shared-memory locations that are not associated with that component. On the one hand, this excludes operations such as $m$-CAS that modify a component conditionally on external metadata matching an expected value. On the other hand, it precludes modifying a component while simultaneously recording information in a shared-memory location that is not associated with that component. In particular, the property prohibits relying on external data when applying the modification and prevents information about the modification from being leaked to external memory locations.

We are now ready to state the first  theorem in this paper. 
\begin{theorem}\label{th:main_oblivious}
A linearizable adaptive snapshot with oblivious click algorithm cannot employ invisible reads.  
\end{theorem}


Theorem~\ref{th:main_oblivious} holds even with only one component and at most three processes. The proof of Theorem~\ref{th:main_oblivious} appears in Section~\ref{sec:first_impossible}. The proof relies on constructing two indistinguishable executions that impose different linearization orders. The formal construction is presented in Subsection~\ref{subsec:first_construction}.

\subsection{Adaptive snapshot with a $k$-bounded click} \label{subsec:non-intrusive}

While the properties introduced above are common and satisfied by all
known implementations, future designs may allow \textsc{Click}
operations to access a limited number of components. We show that our impossibility result holds for such algorithms as well, with somewhat stronger assumptions. To this end, we introduce additional properties and modify the theorem so that it does not rely on
the assumption that \textsc{Click} operations are completely oblivious to the
components. 
We call algorithms satisfying the new set of properties 
\textsl{adaptive snapshots with a $k$-bounded click}. 

As with the previous set of properties, we begin by defining
\textsl{Update linearization independence from \textsc{Click} and \textsc{Invoke}}. This property is
similar to update linearization independence from \textsc{Click}, but additionally
requires that updates to different components be independent. We formally
define it below.

\begin{definition}[Update linearization independence from \textsc{Click} and \textsc{Invoke}]
\label{def:non_intrusive_invoke_independence}
Let $\mathcal{A}$ be an adaptive snapshot algorithm and $E$ be a \execType execution of $\mathcal{A}$ with history
$e_1, e_2, \ldots$ (finite or infinite). Let $p$ be a process that executes an effective linearization step of an updating operation \textsc{Invoke}$(i,op,arg)$. Let $e_t$ denote the step that is the effective linearization step of this operation in $E$. Consider a \execType execution $E'$ whose first $t-1$ events are identical to those of $E$, followed by an arbitrary sequence of steps, not executed by $p$, consisting solely of \textsc{Click} operations, or \textsc{Invoke}$(j,op_2,arg_2)$ operations with $j \neq i$ (either continuing
previously invoked operations or newly invoked ones). After these steps, process $p$ performs step $e_t$, followed by any valid suffix of execution
steps. Then, in execution $E'$, the effective linearization step of
\textsc{Invoke}$(i,op,arg)$ is the event $e_t$, performed by $p$, and the
result of the operation (both the state modification and the returned value)
is identical to its result in execution $E$.

We say that an adaptive snapshot algorithm $\mathcal{A}$ satisfies
\emph{Update linearization independence from \textsc{Click} and \textsc{Invoke}} if the property defined above holds for every \execType execution $E$ of $\mathcal{A}$, every process $p$, and every
\textsc{Invoke}$(i,op,arg)$ operation in $E$.

\end{definition}

\begin{definition}[Adaptive snapshot with a $k$-bounded click]\label{def:properties_non_intrusive}
A linearizable adaptive snapshot algorithm $\mathcal{A}$ that maintains a collection $\mathcal{O}$ of $m$ linearizable components is called an
\textsl{adaptive snapshot with a $k$-bounded click} algorithm if it implements the adaptive snapshot object and satisfies the
following properties:
\begin{enumerate}
    \item \label{prop:limit_memory}
    \emph{$k$-bounded click:}
    Each \textsc{Click} operation accesses no more than $k$ components. 

    \item \label{prop:obstruction_free}
    \emph{Progress guarantee:}
$\mathcal{A}$ is an obstruction-free adaptive snapshot algorithm.

    \item \label{prop:non_intrusive_linearize}
    \emph{Update linearization independence from \textsc{Click} and \textsc{Invoke}:}
    $\mathcal{A}$ satisfies \emph{Update linearization independence from \textsc{Click} and \textsc{Invoke}} as in Definition~\ref{def:non_intrusive_invoke_independence}.

    \item \label{prop:act_on_memory}
    \emph{Component encapsulation:}
    Component~$\mathcal{O}[i]$ is modified only through \textsc{Invoke}$(i,op,arg)$ operations. Moreover, in \execType executions, the effective linearization step of
    an updating operation, accesses and modifies only the portion of shared memory associated with component~$\mathcal{O}[i]$.


    
\end{enumerate}
\end{definition}

Property~\ref{prop:limit_memory} bounds the number of memory locations accessed by a \textsc{Click} operation. 
For appropriate
choices of $k$, this restriction rules out expensive \textsc{Click} operations and, in particular, classical (non-adaptive) snapshot algorithms in which a
snapshot explicitly records the entire memory.

Property~\ref{prop:non_intrusive_linearize} requires that, similarly
to Property~\ref{prop:oblvious_click_linearize} in
Definition~\ref{def:properties_oblivious_click}, steps of
\textsc{Click} operations do not interfere with a pending effective
linearization step of an updating operation. In addition, steps of
\textsc{Invoke} operations on other components must not interfere
with such effective linearization steps either.

The other two properties are similar to those in Definition~\ref{def:properties_oblivious_click}. 

\ignore{

As in the previous set of properties, Property~\ref{prop:obstruction_free}
requires obstruction-freedom, a weak progress guarantee. This excludes
implementations that block all operations in order to initiate a snapshot.


\textcolor{red}{
Property~\ref{prop:non_intrusive_linearize} states that \textsc{Click} operations, as well as \textsc{Invoke} operations on other components, do not interfere with the effective linearization steps of updating operations. Note that such operations may still affect the actual linearization point of an \textsc{Invoke} operation. However, they do not affect the specific \emph{effective} linearization step of the updating \textsc{Invoke} operation, as observed by solo \textsc{Read} operations. }

Finally, Component encapsulation (Property~\ref{prop:act_on_memory}) is as in the previous set of properties.

}

Our second theorem states an impossibility similar to the one in Theorem~\ref{th:main_oblivious} for adaptive snapshots with a $k$-bounded click. 

\begin{theorem}\label{th:main_non_intrusive}
Let $k=\min\{ m-1, n-3 \}$. A linearizable adaptive snapshot with a $k$-bounded click algorithm cannot employ invisible reads.  
\end{theorem}

In Section~\ref{sec:second_impossible}, we prove Theorem~\ref{th:main_non_intrusive}. The proof follows a construction similar to that of the previous case. However, since a \textsc{Click} operation may access up to $k$ components, the construction employs $k+3$ processes and $k+1$ components. The formal construction is presented in Subsection~\ref{subsec:second_construction}.


\section{Adaptive snapshot with oblivious click cannot employ invisible reads} \label{sec:first_impossible}

We begin with an informal discussion. The main challenge in executing a \textsc{Click} concurrently with updates lies in determining which concurrent updates occur before the snapshot and which occur after it. The key idea of the impossibility proof is to show that a concurrent update to a component may linearize either early (before the \textsc{Click}) or late (after the \textsc{Click}) in a manner that is indistinguishable from the perspective of the process executing the \textsc{Click}. Traditionally, \textsc{Click} implementations address this problem by retroactively adjusting the linearization order of operations. However, in the presence of \textsc{Read} operations, the ability to reorder updates, whose effects are observable by reads, is restricted. In particular, an update to a component cannot be moved past a \textsc{Read} that observes either the updated state or the preceding state of the same component. Snapshot algorithms typically overcome the witnessing of update timing by reads by making readers help earlier updates resolve their linearization order relative to the \textsc{Click}. However, when reads are invisible, they are not permitted to write to shared memory, making it impossible to linearize the \textsc{Click} operation correctly.  We formalize this intuition in the proof below.

Let $p$ be a process that executes a \textsc{Click}. We use the \emph{Click-obliviousness} property (Property~\ref{prop:limit_memory} of Definition~\ref{def:properties_oblivious_click}) to deduce that the state of any component cannot be read by $p$ during the execution of the \textsc{Click} operation. Consequently, the \textsc{Click} is oblivious to whether an effective linearization step of a concurrent update occurs  before or after its linearization point. Thus, we show that this results in two possible linearization orders that are indistinguishable with respect to process~$p$.

To exploit this observation, we construct two executions in which the relative order of a \textsc{Click} and a concurrent update differs, although the processes that must determine this order have indistinguishable views. Concurrent invisible \textsc{Read} operations fix the relative order of the operations but cannot report the order to other processes because they do not modify shared memory. Thus, we obtain that one of these two executions must yield a non-linearizable execution. 
Otherwise, subsequent \textsc{Observe} operations must return different values in the two executions, leading to a contradiction. We now formalize these ideas. 

\subsection{Construction of the two indistinguishable executions } \label{subsec:first_construction}

Assume, by way of contradiction, that there exists a linearizable adaptive snapshot with oblivious click algorithm $\mathcal{A}$ (see Definition~\ref{def:properties_oblivious_click}) that implements invisible \textsc{Read} operations and maintains a collection $\mathcal{O}$ of $m$ linearizable components. We begin by constructing two executions, denoted $E$ and $E'$, which impose different linearization orders between one updating operation and a \textsc{Click}, implying that the snapshot obtained must differ between the two executions. However, we show that the two executions are indistinguishable with respect to a process that executes a \textsc{Click} and an \textsc{Observe}. This implies that the returned value from the \textsc{Observe} operation must be erroneous in one of these executions, yielding a contradiction.

Recall that $m$ denotes the number of components in the snapshot object and $n$ denotes the number of processes in the system, labeled $p_1,\ldots,p_n$. 
We know that  $\mathcal{A}$ is linearizable, and that all operations of $\mathcal{A}$ and all component operations are obstruction-free (Property~\ref{prop:oblivious_obstruction_free} of Definition~\ref{def:properties_oblivious_click}). Executions $E$ and $E'$ are constructed so that, in each execution, there is exactly one \textsc{Invoke} operation, which is an updating operation, and it is invoked on component $\mathcal{O}[1]$. Therefore, $E$ and $E'$ are \execType executions, and every completed updating operation in these executions has an effective linearization step (see Observation~\ref{obs:elp_exist}).

We construct the executions using one process, $p_1$, to perform a single updating \textsc{Invoke}$(1,op,arg)$ operation. We additionally use two processes: process $p_{read}$ performs invisible \textsc{Read} operations, and process $p_{click}$ performs a single \textsc{Click} operation. Consequently, if the \textsc{Click} operation is invoked before the linearization point of this
\textsc{Invoke} operation, then a subsequent \textsc{Observe}$(1)$ operation must return the initial state of $\mathcal{O}[1]$. In contrast, if a \textsc{Click} operation is invoked after the linearization point of this
\textsc{Invoke} operation, then a subsequent \textsc{Observe}$(1)$ operation must return a different state of $\mathcal{O}[1]$.

We consider a legal initial configuration in which every component $\mathcal{O}[i]$, for $i\in[m]$, is in the same state $o_1$. It suffices to establish the impossibility from such a configuration. Let $op$ be an update operation and let $arg$ be a parameter such that applying $op$ with parameter $arg$ to a component in state $o_1$ changes its state. Since all components are instances of the same object type, this update has the same effect on every component in state $o_1$. We denote the resulting state by $o_2$, where $o_2\neq o_1$.

We will focus on component~$\mathcal{O}[1]$. We now construct an initial execution $E''$, which will be extended to obtain $E$ and $E'$. In $E''$, processes $p_1$ and $p_{read}$ alternate: $p_1$ executes an \textsc{Invoke}$(1,op,arg)$ operation one step at a time, and between consecutive steps of this operation, $p_2$ executes a complete \textsc{Read}$(1)$ operation. Formally, $E''$ is defined as follows.

\begin{enumerate}
\item \label{step:oblivious_invoke} Processes $p_{read}$ and $p_1$ are alternating in the execution in the following manner. 
\begin{enumerate}
\item \label{step:oblivious_read} Process $p_{read}$ executes a  \textsc{Read}$(1)$ operation to completion. If the returned state is $o_2$, the construction of $E''$ is completed.
\item \label{step:oblivious_next_event} Process $p_1$ performs one step of its \textsc{Invoke}$(1,op,arg)$ operation. Once the operation completes, $p_1$ takes no further steps.
\end{enumerate}
\end{enumerate}

We first prove that the construction of $E''$ completes and yields a  well-defined execution.  We note that the operations that are invoked are operations of the snapshot object with adequate parameters, executed concurrently. It remains to show that the construction completes. Informally, using obstruction-freedom, the updating operation must complete when run in isolation, and invisible reads cannot change that. Therefore, the component state must change to $o_2$ after finitely many steps, and the construction ends when the \textsc{Read} operation observes the resulting state. This intuition is formalized in Lemma~\ref{lemma:execution_valid}.

\begin{lemma} \label{lemma:execution_valid}
    Execution $E''$ is a valid finite execution of Algorithm~$\mathcal{A}$.
\end{lemma}

\begin{proof}
For lack of space, the proof is postponed to Appendix~\ref{proof:execution_valid}.
\end{proof}

Consider the final iteration of the construction. In this iteration, \textsc{Read}$(1)$ returns $o_2$, whereas every preceding \textsc{Read}$(1)$ operation returns a different state. Additionally, \textsc{Read} operations performed by $p_{read}$ are invisible and do not modify the shared memory. Consequently, by Definition~\ref{def:effective_linearization}, the final step taken
by $p_1$ in $E''$ must be the effective linearization step of the
\textsc{Invoke} operation.

Let $C_0$ be the initial configuration, $Q$ be the configuration before the last step of $p_1$ in $E''$, and $Q_0$ be the configuration after the execution of this step. 
We now construct two execution extensions, denoted $E$ and $E'$, starting at
the configurations $Q_0$ and $Q$, respectively. The executions $E$ and $E'$ are depicted in Figure~\ref{fig:oblivious_diagram}.

Execution $E$ linearizes the \textsc{Click} after the updating operation. Starting at configuration $Q_0$, process $p_{read}$ executes \textsc{Read}$(1)$ to completion, after which process $p_{click}$ executes \textsc{Click}$()$ to completion. Execution $E'$ linearizes the \textsc{Click} before the update. Starting at configuration $Q$, process $p_{read}$ executes \textsc{Read}$(1)$ to completion, process $p_{click}$ then executes \textsc{Click}$()$ to completion, and $p_{read}$ executes another \textsc{Read}$(1)$ to completion. Finally, $p_1$ takes one step.

Since both extensions only add operations that execute from start to end with no concurrent interruptions, and by obstruction-freedom, both extensions are well-defined and terminate. The executions $E$ and $E'$ are similar except for the relative order between the \textsc{Click} operation, the invisible \textsc{Read}$(1)$ operation by $p_{read}$, and the step of $p_1$. This step constitutes the effective linearization step of the updating \textsc{Invoke} operation in execution $E$ by its definition, and we will show that it constitutes the effective linearization step of the updating \textsc{Invoke} operation in Execution $E'$ as well.    

\begin{figure}[!htbp]
\centering
\begin{tikzpicture}[
    x=1.1cm,
    y=1cm,
    timeline/.style={thick},
    cfg/.style={circle,fill=black,inner sep=1.5pt},
    execstep/.style={-Stealth,thick},
    sim/.style={dashed,thick}
]


\tikzset{sim/.style={draw=none}} 

\def\yE{0}
\def\yEp{-1.5}

\draw[timeline] (0.5,\yE) -- (12,\yE);
\draw[timeline] (0.5,\yEp) -- (12,\yEp);

\node[left] at (0.5,\yE) {$E$};
\node[left] at (0.5,\yEp) {$E'$};

\node[cfg,label=above:$C_0$] (C0) at (0.9,\yE) {};
\node[cfg,label=below:$C_0$] (C0p) at (0.9,\yEp) {};
\draw[sim] (C0) -- (C0p);

\node[cfg,label=above:$Q$] (C) at (1.5,\yE) {};
\node[cfg,label=below:$Q$] (Cp) at (1.5,\yEp) {};
\draw[sim] (C) -- (Cp);

\node[cfg,label=above:$Q_0$] (D) at (3.5,\yE) {};
\draw[execstep](C) --
  node[above, align=center, yshift=0pt]
  {\small
   $p_1$ linearizes\\
   \small the update }
  (D);

\draw[sim] (D) -- (Cp);

\node[cfg,label=above:$Q_1$] (Q1) at (5.75,\yE) {};
\draw[execstep](D) --
  node[above, align=center, yshift=0pt]
  {\small
   $p_{read}$ executes\\
   \small \textsc{Read}$(1)$}
  (Q1);

\node[cfg,label=above:$C$] (Q2) at (8,\yE) {};
\draw[execstep](Q1) --
  node[above, align=center, yshift=0pt]
  {\small
   $p_{click}$ executes\\
   \small \textsc{Click} }
  (Q2);

\node[cfg,label=below:$Q'_1$] (Q1p) at (5.75,\yEp) {};
\draw[execstep](Cp) --
  node[above, align=center, yshift=0pt]
  {\small
   $p_{read}$ executes\\
   \small \textsc{Read}$(1)$}
  (Q1p);

\node[cfg,label=below:$Q'_2$] (Q2p) at (8,\yEp) {};
\draw[execstep](Q1p) --
  node[above, align=center, yshift=0pt]
  {\small
   $p_{click}$ executes\\
   \small \textsc{Click} }
  (Q2p);

\draw[sim] (Q1) -- (Q1p);

\draw[sim] (Q2) -- (Q2p);

\node[cfg,label=below:$Q'_3$] (Q3p) at (10,\yEp) {};
\draw[execstep](Q2p) --
  node[above, align=center, yshift=0pt]
  {\small
   $p_{read}$ executes\\
   \small \textsc{Read}$(1)$}
  (Q3p);

\draw[sim] (Q2) -- (Q3p);

\node[cfg,label=below:$C'$] (Q4p) at (12,\yEp) {};
\draw[execstep](Q3p) --
  node[above, align=center, yshift=0pt]
  {\small
   $p_1$ linearizes\\
   \small the update }
  (Q4p);

\draw[sim] (Q2) -- (Q4p);


\end{tikzpicture}

\caption{A depiction of executions $E$ and $E'$.} \label{fig:oblivious_diagram}
\end{figure}
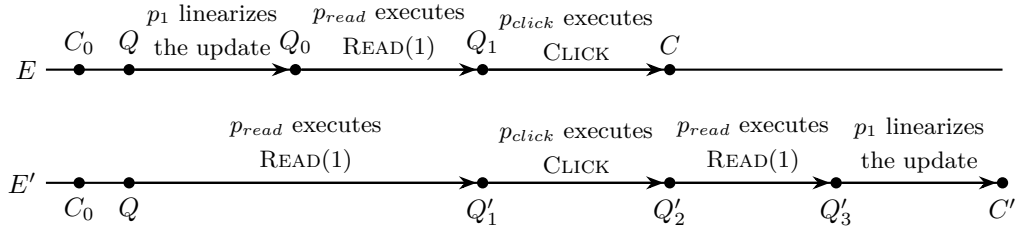


Let $Q_1$ and $Q_1'$ denote the configurations reached in $E$ and $E'$, respectively, after $p_{read}$ completes the first $\textsc{Read}(1)$ operation that is not part of $E''$.  Let $C$ and $Q_2'$ denote the configurations
reached after $p_{click}$ completes a \textsc{Click} operation starting
from $Q_1$ and $Q_1'$, respectively. Let $Q_3'$ denote the
configuration reached after $p_{read}$ completes a $\textsc{Read}(1)$
operation starting at $Q_2'$, and finally let $C'$ denote the
configuration reached after $p_1$ takes a single step from $Q_3'$.

We denote the $\textsc{Invoke}(1,op,arg)$ operations in executions $E$
and $E'$ by $op_1$ and $op_1'$, respectively. The $\textsc{Read}(1)$
operation executed from configuration $Q_0$ in $E$ is denoted by $op_2$, and
the $\textsc{Read}(1)$ operation executed from configuration $Q_2'$ in $E'$ is
denoted by $op_2'$. Finally, we denote the \textsc{Click} operations
in $E$ and $E'$ by $op$ and $op'$, respectively.

We now claim that the two
configurations $C$ and $C'$ are indistinguishable to process $p_{click}$, and that the
shared-memory state is identical in both configurations. This claim is
proved in the lemma below.

\begin{lemma} \label{lemma:oblivious_conf_similar}
Let $C$ and $C'$ be the configurations defined above. Then
$C \isim{p_{click}} C'$, and the shared-memory state in both configurations
is identical.
\end{lemma}

\ignore{

\begin{proof}
Informally, the memory accesses of the \textsc{Click} and the effective linearization step of the \textsc{Invoke} operation do not access the same portion of shared memory, given the assumptions we made. This implies indistinguishability. Also, for the same reason, modifications of the shared memory by the \textsc{Click} and by the effective linearization step of the updating operation are independent and their order does not impact the final state of the shared space.
Due to lack of space, the formal proof is postponed to Appendix~\ref{proof:oblivious_conf_similar}.
\end{proof}

}

\begin{proof}
We begin by comparing configurations $Q$ and $Q_0$. As explained earlier,
the transition from $Q$ to $Q_0$ in execution $E$ consists of a single step taken by
$p_1$, in which it performs the effective linearization step of the
$\textsc{Invoke}(1,op,arg)$ operation. 

By Property~\ref{prop:oblivious_act_on_memory} of Definition~\ref{def:properties_oblivious_click}, this step does not modify
any shared-memory locations except possibly, shared memory associated with $\mathcal{O}[1]$. Hence, for process $p_{click}$, which has not yet taken
part in the execution and has not yet accessed the shared memory, the local state is identical in $Q$ and $Q_0$
(i.e., $Q_0 \isim{p_{click}} Q$), and the shared memory may differ only in shared-memory locations associated with~$\mathcal{O}[1]$.

Then, configurations $Q_1$ and $Q_1'$ are obtained by
$p_{read}$ completing a $\textsc{Read}(1)$ operation, starting at configurations $Q_0$ and $Q$, respectively. Since reads are
invisible in $\mathcal{A}$, $p_{read}$ does not modify shared memory during
this step, and $p_{click}$ does not change its local state, since it does not perform a step. Therefore,
$Q_1 \isim{p_{click}} Q_1'$, and the shared-memory states may differ only in shared-memory locations associated with~$\mathcal{O}[1]$.

Next, configurations $C$ and $Q_2'$ are obtained by having $p_{click}$ complete a \textsc{Click} operation, starting at configurations $Q_1$ and $Q_1'$, respectively. By Property~\ref{prop:oblivious_limit_click} of Definition~\ref{def:properties_oblivious_click}, $p_{click}$ does not access $\mathcal{O}[1]$ during the execution of this operation in $E$. Therefore, all shared-memory locations accessed during the \textsc{Click} execution have identical values in $Q_1$ and $Q'_1$.

By Lemma~\ref{lemma:indistinguishable_configurations}, executing the \textsc{Click} operation from either configuration results in identical modifications to the shared memory and leaves $p_{click}$ in the same local state at the end of both executions. Consequently, $C \isim{p_{click}} Q_2'$, and the shared-memory states in the two configurations may differ only in shared-memory locations associated with~$\mathcal{O}[1]$.

From $Q_2'$, configuration $Q_3'$ is obtained by $p_{read}$ completing another $\textsc{Read}(1)$ operation. Since reads are invisible, shared memory is unchanged and $p_{click}$’s local state remains the same. Therefore, $C \isim{p_{click}} Q_3'$, and again the shared-memory states may differ only in shared-memory locations associated with~$\mathcal{O}[1]$.

Finally, configuration $C'$ is obtained from $Q_3'$ by a single step of $p_1$. Let $e$ denote the last step before configuration $Q_0$ in execution $E$. The prefixes of $E$ and $E'$ up to step $e$ are identical, and $e$ is the effective linearization step of the $\textsc{Invoke}(1,op,arg)$ operation in $E$. Observe that $Q_3'$ is obtained from $Q$, the configuration immediately preceding $e$ in $E$, by executing only \textsc{Click} and \textsc{Read} steps, with no process executing steps of both types. Additionally, in $E$, there is one \textsc{Invoke} operation, which is updating. By
Lemma~\ref{lemma:invoke_independence_to_read}, the next step of $p_1$ from $Q_3'$ is the effective linearization step of the
\textsc{Invoke}$(1,op,arg)$ operation in $E'$, and it modifies the shared-memory locations associated with~$\mathcal{O}[1]$ in exactly the same way as $e$ does when executed from $Q$. Since the difference between the shared-memory states of the configurations $Q_3'$ and $C$ is precisely the effect of step $e$ on the shared memory associated with $\mathcal{O}[1]$, the configuration obtained after $p_1$ performs $e$ from $Q'_3$, namely $C'$, has the same shared-memory state as $C$. Moreover, $p_{click}$ does not take any steps between $Q_3'$ and $C'$, and since $Q_3' \isim{p_{click}} C$, it follows that $C' \isim{p_{click}} C$.
\end{proof}

We next establish the opposite relative order of the linearization points of $op$ and $op_1$ in $E$, and of $op'$ and $op_1'$ in $E'$.

\begin{lemma} \label{lemma:order_of_linearization}
In execution $E$, the linearization point of $op_1$ (the \textsc{Invoke}) occurs before the linearization point of $op$ (the \textsc{Click}), whereas in execution $E'$, the linearization point of $op_1'$ (the \textsc{Invoke}) occurs after the linearization point of $op'$ (the \textsc{Click}).
\end{lemma}

\begin{proof}
In execution $E$, process $p_{read}$ executes a \textsc{Read}$(1)$ (denoted $op_2$) operation from configuration $Q_0$, when the effective linearization step of $op_1$ has already occurred. Therefore, $op_2$ must return the updated state, $o_2$, and the linearization point of $op_2$ occurs after the linearization point of $op_1$. Since $op_2$ and $op$ are not concurrent and $op_2$ precedes $op$ in the execution $E$, then the linearization point of $op$ occurs after the linearization point of $op_2$. By transitivity, it follows that the linearization point of $op$ occurs after the linearization point of $op_1$.

In execution $E'$, operation $op'$ strictly precedes the \textsc{Read}$(1)$ operation that follows it (from configuration $Q_2'$, the operation denoted by $op_2'$). Thus, the linearization point of $op'$ occurs before the linearization point of $op_2'$.

Execution $E'$ is obtained from $E''$ by removing the step that constitutes the effective linearization step, and then extending the execution with steps of \textsc{Read} and \textsc{Click} operations, such that no process executes steps from both \textsc{Click} and \textsc{Read} operations, before allowing $p_1$ to take its next step. This construction exactly matches the conditions in  Lemma~\ref{lemma:invoke_independence_to_read}. Hence, by Lemma~\ref{lemma:invoke_independence_to_read}, the step taken by $p_1$ from configuration $Q_3'$ to $C'$ is the effective linearization step of the \textsc{Invoke} operation in execution $E'$.

By the definition of an effective linearization step, every \textsc{Read}$(1)$ operation that completes before this step, including $op_2'$, returns the state preceding the \textsc{Invoke}. Therefore, the linearization point of $op_2'$ must occur before the linearization point of $op_1'$. By transitivity, the linearization point of $op_1'$ occurs after the linearization point of $op'$.

\end{proof}

By Lemma~\ref{lemma:order_of_linearization}, the snapshot contains different states of $\mathcal{O}[1]$ in the two executions. A subsequent \textsc{Observe} operation must reveal this difference.

By Lemma~\ref{lemma:oblivious_conf_similar}, $C\isim{p_{click}}C'$ and the two configurations have identical shared-memory states. Lemma~\ref{lemma:indistinguishable_configurations} therefore implies that any sequence of events $\sigma$ consisting solely of steps by process $p_{click}$ that can occur starting at $C$ can also occur starting at $C'$. Moreover, if $\sigma$ is finite, then $C\sigma \isim{p_{click}} C'\sigma$.

Let $\sigma$ be the sequence of events in which $p_{{click}}$ executes \textsc{Observe}$(1)$ in isolation starting at configuration~$C$. Since $\mathcal{A}$ is obstruction-free (Property~\ref{prop:oblivious_obstruction_free} of Definition~\ref{def:properties_oblivious_click}), the execution of $\sigma$ starting at $C$ must terminate and return a value which, by Lemma~\ref{lemma:order_of_linearization}, must be $o_2$.

The execution from $C'$ must also terminate and, by Lemma~\ref{lemma:order_of_linearization}, return $o_1$. However, responses are part of the local state of $p_{click}$, and Lemma~\ref{lemma:indistinguishable_configurations} gives $C\sigma \isim{p_{click}}C'\sigma$. Thus, the two \textsc{Observe} operations cannot return different values, a contradiction.

Therefore, our initial assumption (by way of contradiction) is false, and there is no linearizable adaptive snapshot with oblivious click algorithm that can employ invisible \textsc{Read} operations. This concludes the proof of Theorem~\ref{th:main_oblivious}.


\section{Adaptive snapshot with a $k$-bounded click cannot employ invisible reads} \label{sec:second_impossible}

As in the previous section, we begin with an informal discussion. Let $p$ be a process that executes a \textsc{Click}. We use the \emph{k-bounded click} property (Property~\ref{prop:limit_memory} of Definition~\ref{def:properties_non_intrusive}) to deduce that in the presence of $k+1$ concurrent \textsc{Invoke} operations, each on a different component, the state of at least one of these components cannot be read by $p$ during the execution of the \textsc{Click} operation. Consequently, the \textsc{Click} operation is oblivious to whether the update to such a component is linearized before or after the \textsc{Click}. Thus, the two possible orders are indistinguishable to $p$.

As in the previous section, we construct two executions with different relative orders between a \textsc{Click} and a concurrent update, while the processes responsible for determining this order have indistinguishable views. Invisible \textsc{Read} operations fix the relative order but cannot report it to the concurrent operations because they do not modify shared memory. Therefore, one of these two executions must yield a non-linearizable execution. 
Otherwise, subsequent \textsc{Observe} operations must return different values in the two executions, leading to a contradiction. We now proceed to formalize the construction.

\subsection{Construction of the two indistinguishable executions } \label{subsec:second_construction}

Let $k=\min\{ m-1, n-3 \}$. Assume, by way of contradiction, that there exists a linearizable adaptive snapshot with a $k$-bounded click algorithm~$\mathcal{A}$ (see Definition~\ref{def:properties_non_intrusive}) that maintains a collection $\mathcal{O}$ of $m$ linearizable components, and whose \textsc{Read} operations are invisible. We construct two executions, $E$ and $E'$, with different linearization orders whose final configurations are indistinguishable to one process. This yields a contradiction and proves that no such algorithm $\mathcal{A}$ exists.

Recall that $m$ denotes the number of components in the snapshot object and $n$ denotes the number of processes in the system, labeled $p_1,\ldots,p_n$. We know that $\mathcal{A}$ is linearizable, and that all operation of $\mathcal{A}$ and all operation of the components are obstruction-free (Property~\ref{prop:obstruction_free} of Definition~\ref{def:properties_non_intrusive}). The executions are constructed so that, in each execution, there is at most one updating operation on each component. Thus, $E$ and $E'$ are component-serialized, and every completed updating operation in them has an effective linearization step (see Observation~\ref{obs:elp_exist}).

Before constructing the executions, we first state a simple lemma. The lemma states that if an adaptive snapshot with a $k$-bounded click algorithm supports invisible \textsc{Read} operations, then update linearization is independent of clicks, reads, and updates to other components. The lemma is stated formally below.

\begin{lemma}[Update linearization independence from \textsc{Click}, \textsc{Read}, and \textsc{Invoke}]
\label{lemma:invoke_independence_to_read_k_bound}

Let $\mathcal{A}$ be a linearizable adaptive snapshot with a $k$-bounded click algorithm that implements invisible \textsc{Read} operations, and $E$ be a \execType execution of $\mathcal{A}$ with history
$e_1, e_2, \ldots$ (finite or infinite). Let $p$ be a process that executes the effective linearization step of an updating operation \textsc{Invoke}$(i,op,arg)$. Let $e_t$ denote the step that is the effective linearization step of this operation in $E$. Consider a \execType execution $E'$, whose first $t-1$ events are identical to the first $t-1$ events of $E$, followed by an arbitrary sequence of steps, none of which are executed by $p$, consisting solely of \textsc{Click} and \textsc{Read} operations, or \textsc{Invoke}$(j,op_2,arg_2)$ operations with $j \neq i$ (either continuing previously invoked operations or newly invoked ones), such that no process executes steps of both \textsc{Read} and \textsc{Click}, or both \textsc{Read} and \textsc{Invoke}, within that sequence. After these steps, process $p$ performs step $e_t$, followed by any valid suffix of execution steps. Then, in execution $E'$, the effective linearization step of \textsc{Invoke}$(i,op,arg)$ is the event $e_t$ performed by $p$, and the result of the operation (both the state modification and the returned value) is identical to its result in execution $E$.
\end{lemma}

\begin{proof}
The proof appears in  Appendix~\ref{proof:invoke_independence_to_read_k_bound}.
\end{proof}

We construct the executions using $k+1$ processes, $p_1, \ldots, p_{k+1}$, each performing a single updating 
\textsc{Invoke}$(i,op,arg)$ operation. $p_{read}$ will execute \textsc{Read} operations in isolation, and $p_{click}$ will execute a \textsc{Click} operation in isolation at a specified point in the construction. Consequently,
if a \textsc{Click} operation is invoked before the linearization point of an
\textsc{Invoke}$(i,op,arg)$ operation, then a subsequent \textsc{Observe}$(i)$ operation must return the initial state of $\mathcal{O}[i]$. In contrast, if a \textsc{Click} operation is invoked after the linearization point of this
\textsc{Invoke} operation, then a subsequent \textsc{Observe}$(i)$ operation must return a different state of $\mathcal{O}[i]$.

We consider a legal initial configuration in which every component $\mathcal{O}[i]$, for $i\in[m]$, is in the same state $o_1$. It suffices to establish the impossibility from such a configuration. Let $op$ be an update operation and let $arg$ be a parameter such that applying $op$ with parameter $arg$ to a component in state $o_1$ changes its state. Since all components are instances of the same object type, this update has the same effect on every component in state $o_1$. We denote the resulting state by $o_2$, where $o_2\neq o_1$.

We now define, for each $i\in[k+1]$, an execution scheme denoted by $E''(i)$. The notation $E''(i)$ does not refer to a single execution from a fixed initial configuration. Rather, the scheme can be instantiated from any configuration $C$, in which case it specifies an execution segment that can occur from $C$. The executions $E$ and $E'$ constructed later will use different instantiations of this scheme. An instantiation of $E''(i)$ proceeds as follows.

\begin{enumerate}
\item Processes $p_{read}$ and $p_i$ alternate in the execution in the following manner. 
\begin{enumerate}
\item \label{step:read_i_k_bounded} Process $p_{read}$ executes a \textsc{Read}$(i)$ operation to completion. If the returned state is $o_2$, the construction of $E''(i)$ is completed.
\item \label{step:next_event_i_k_bounded} Process $p_i$ performs one step of its \textsc{Invoke}$(i,op,arg)$ operation. Once the operation completes, $p_i$ takes no further steps.
\end{enumerate}
\end{enumerate}

We first observe that, for every $1 \leq i \leq k+1$, the execution scheme $E''(i)$ is obtained from the construction of $E''$ in Subsection~\ref{subsec:first_construction} by replacing $p_1$ and $\mathcal{O}[1]$ with $p_i$ and $\mathcal{O}[i]$, respectively. The proof of Lemma~\ref{lemma:execution_valid} applies verbatim after these substitutions. Consequently, when instantiated from any configuration, $E''(i)$ specifies a valid finite execution segment.

Having established that, for every $1 \leq i \leq k+1$, the scheme $E''(i)$ can be instantiated at any configuration to yield a valid finite execution segment, we now construct the execution $E''$. Our goal is to ensure that, in the final configuration of $E''$, the next step of each process $p_1,\ldots,p_{k+1}$ is the effective linearization step of the corresponding \textsc{Invoke} operation. We construct $E''$ inductively by successively instantiating the schemes $E''(1),\ldots,E''(k+1)$.

Formally, we construct $E''$ beginning from the initial configuration
$C_0$. For every $1 \le i \le k+1$, we extend the execution from
configuration $C_{i-1}$ to $C_i$ by a sequence of events $\sigma_i$
such that:
(i) $\sigma_i$ is $\{p_i,p_{read}\}$-only, and
(ii) in configuration $C_i$, the next step of process $p_i$ is the
effective linearization step of an
\textsc{Invoke}$(i,op,arg)$ operation.

\paragraph*{Base case ($i=1$).}
Consider an instantiation of the scheme $E''(1)$. Since we have shown that this instantiation terminates, there exists an iteration $j \ge 2$ such that the
\textsc{Read}$(1)$ operation in
Stage~\ref{step:read_i_k_bounded} (see the definition of execution scheme $E''(i)$ above) returns the updated value $o_2$,
while the \textsc{Read}$(1)$ in the preceding iteration returns the initial state $o_1$.

Note that $j \neq 1$, because prior to the first
\textsc{Read}$(1)$ operation, the
\textsc{Invoke}$(1,op,arg)$ operation has not yet begun.
Hence, the modification to $\mathcal{O}[1]$ cannot occur before
Stage~\ref{step:next_event_i_k_bounded} in iteration~1.

We therefore define $C_1$ to be the configuration immediately after
Stage~\ref{step:read_i_k_bounded} in iteration $j-1$.
The extension from $C_0$ to $C_1$ is $\{p_1,p_{read}\}$-only, since
these are the only processes that take steps in $E''(1)$.
Moreover, because $p_{read}$ performs only invisible
\textsc{Read} operations, the effective linearization step of
\textsc{Invoke}$(1,op,arg)$ is executed by $p_1$, and it occurs in
Stage~\ref{step:next_event_i_k_bounded} of iteration~$j$.
Thus, in configuration $C_1$, the next step of $p_1$ is precisely
its effective linearization step.

\paragraph*{Induction step.}
Let $1 \le i \le k$, and assume that we have constructed $E''$
up to configuration $C_i$ that satisfies the induction hypothesis.
From $C_i$, we extend the execution by instantiating $E''(i+1)$.

As in the base case, let $j$ denote the iteration in which
the instantiation of $E''(i+1)$ terminates, and define $C_{i+1}$ to be the configuration
immediately after Stage~\ref{step:read_i_k_bounded} in iteration
$j-1$. By the same reasoning as before, the extension from $C_i$
to $C_{i+1}$ is $\{p_{i+1},p_{read}\}$-only, since these are the only
processes that take steps in $E''(i+1)$.
Furthermore, from configuration $C_{i+1}$, the next step of
$p_{i+1}$ is the effective linearization step of
\textsc{Invoke}$(i+1,op,arg)$.

This completes the induction and the construction of $E''$.

\paragraph*{Construction of $E$ from $E''$.}
Having constructed $E''$, we now define execution $E$.
Let $C_{k+1}$ denote the final configuration of $E''$.
From $C_{k+1}$, each process $p_1,\ldots,p_{k+1}$ performs exactly one step. Next, process $p_{read}$ performs a
\textsc{Read}$(i)$ operation to completion for every
$1 \le i \le k+1$.
Finally, process $p_{click}$ performs a
\textsc{Click} operation to completion.
The resulting configuration marks the end of execution $E$.

\paragraph*{Construction of $E'$ from $E$.}
Having constructed $E$, we now define execution $E'$.
By Property~\ref{prop:limit_memory} of
Definition~\ref{def:properties_non_intrusive}, a
\textsc{Click} operation in $\mathcal{A}$ accesses no more than $k$
components. Hence, there exists an index
$j \in \{1,\ldots,k+1\}$ such that the component
$\mathcal{O}[j]$ is not accessed during the
\textsc{Click} operation in $E$. We construct $E'$ as follows.
Let $C_{k+1}$ be the final configuration of $E''$.
From $C_{k+1}$, every process $p_1,\ldots,p_{k+1}$,
except $p_j$, performs exactly one step. Next, process $p_{read}$ performs a
\textsc{Read}$(i)$ operation to completion for every
$1 \le i \le k+1$. Then process $p_{click}$ performs a
\textsc{Click} operation to completion. After that, $p_{read}$ again performs a
\textsc{Read}$(i)$ operation to completion for every
$1 \le i \le k+1$. Finally, process $p_j$ performs one step.
The resulting configuration marks the end of execution $E'$.

the two extensions defining $E$ and $E'$ add only a finite number of operations, each of which
executes from invocation to completion without concurrent
interference. By obstruction-freedom, both extensions are therefore
well-defined and must terminate.

Moreover, for each $1\leq i\leq k$, each execution contains exactly one \textsc{Invoke} operation that modifies component~$\mathcal{O}[i]$. Hence, each of these operations is updating. We now introduce notation for the configurations reached in the two
executions.

\paragraph*{Configuration notation.}
Both executions are extended from the same configuration $C_{k+1}$.
For simplicity, we denote $C_{k+1}$ by $Q_0$. For every $1 \le i < j$, let $Q_i$ denote the configuration
reached in both $E$ and $E'$ after process $p_i$ performs its
additional single step starting at $Q_0$.
(These prefixes of $E$ and $E'$ are identical up to this point.) In execution $E$, let $Q_j$ denote the configuration reached
after $p_j$ performs its additional single step. For every $j < i \le k+1$, let $Q_i$ and $Q_i'$ denote the
configurations reached in $E$ and $E'$, respectively,
after $p_i$ performs its additional single step.
Observe that in $E'$, process $p_j$ does not take its step at this
stage; thus, in $E'$, from configuration $Q_{j-1}$ the execution
proceeds directly to $Q_{j+1}'$.

Next, let $D_1$ and $D_1'$ denote the configurations reached in
$E$ and $E'$, respectively, after process $p_{read}$ completes
\textsc{Read}$(i)$ for every $1 \le i \le k+1$, prior to the
\textsc{Click} operation. Let $C$ and $D_2'$ denote the configurations reached after
$p_{click}$ completes a \textsc{Click} operation starting at
$D_1$ and $D_1'$, respectively. Let $D_3'$ denote the configuration reached after
$p_{read}$ completes \textsc{Read}$(i)$ for every
$1 \le i \le k+1$, starting at $D_2'$. Finally, let $C'$ denote the configuration reached after
process $p_j$ performs its single step from $D_3'$.

We denote the \textsc{Invoke}$(j,op,arg)$ operations in executions
$E$ and $E'$ by $op_1$ and $op_1'$, respectively.
Let $op_2$ denote the \textsc{Read}$(j)$ operation executed from
configuration $Q_k$ in $E$, and let $op_2'$ denote the
\textsc{Read}$(j)$ operation executed from configuration $D_2'$ in
$E'$.
Finally, we denote the \textsc{Click} operations in $E$ and $E'$ by
$op$ and $op'$, respectively.

We now claim that configurations $C$ and $C'$ are
indistinguishable to process $p_{click}$, and that the shared memory
state is identical in both configurations.
This claim is formalized and proved in the lemma below.

\begin{lemma} \label{lemma:k_bounded_conf_similar}
Let $C$ and $C'$ be the configurations defined above. Then
$C \isim{p_{click}} C'$, and the shared-memory state in both configurations
is identical.
\end{lemma}

\begin{proof}
Informally, under our assumptions, the locations accessed by the \textsc{Click} operation and by the effective linearization steps of the other updates are disjoint from those accessed by the effective linearization step of \textsc{Invoke}$(j,op,arg)$. This yields indistinguishability. Their shared-memory modifications are therefore independent, and their order does not affect the final shared-memory state. The formal proof appears in Appendix~\ref{proof:k_bounded_conf_similar}.
\end{proof}

We next establish the opposite relative order of the linearization points of $op$ and $op_1$ in $E$, and of $op'$ and $op_1'$ in $E'$.

\begin{lemma} \label{lemma:k_bounded_order_of_linearization}
In execution $E$, the linearization point of $op_1$ (the \textsc{Invoke}) occurs before the linearization point of $op$ (the \textsc{Click}), whereas in execution $E'$, the linearization point of $op_1'$ (the \textsc{Invoke}) occurs after the linearization point of $op'$ (the \textsc{Click}).
\end{lemma}

\begin{proof}
We first analyze execution $E$.

In $E$, operation $op_2$ strictly precedes the \textsc{Click}
operation $op$. Hence, the linearization point of $op$ occurs
after the linearization point of $op_2$.

Recall that $E$ is obtained by extending $E''$.
In the $j$-th stage of the construction of $E''$, the execution
reaches a configuration $C_j$ such that the next step of process
$p_j$ is the effective linearization step of the
\textsc{Invoke}$(j,op,arg)$ operation.

In execution $E$, from configuration $C_j$ until configuration
$Q_{j-1}$, every step taken is either:
(i) a step of a \textsc{Read} operation, or
(ii) a step of an \textsc{Invoke}$(i,op_2,arg_2)$ operation for
$i \neq j$.
Moreover, only process $p_{read}$ performs steps belonging to
\textsc{Read} operations, and it performs no
\textsc{Invoke} or \textsc{Click} steps. Additionally, there is at most one concurrent \textsc{Invoke} operation per component in $E$.
Thus, the conditions of
Lemma~\ref{lemma:invoke_independence_to_read_k_bound}
are satisfied.

It follows that the next step of $p_j$ from configuration
$Q_{j-1}$ is still the effective linearization step of
\textsc{Invoke}$(j,op,arg)$.
By construction, when $p_j$ takes this step, the execution
reaches configuration $Q_j$.

Since there are no other updating operations applied to
$\mathcal{O}[j]$, every \textsc{Read}$(j)$ operation that
completes after this step, in particular $op_2$, which completes
after configuration $Q_{k+1}$, returns the updated value.
Therefore, the linearization point of $op_2$ must occur after the
linearization point of $op_1$.
By transitivity, the linearization point of $op_1$
occurs before the linearization point of $op$.

\medskip
We now analyze execution $E'$.

In $E'$, the \textsc{Click} operation $op'$ strictly precedes
the \textsc{Read}$(j)$ operation $op_2'$.
Hence, the linearization point of $op_2'$ occurs after the
linearization point of $op'$.

As in $E$, the construction of $E'$ extends $E''$.
Again, in the $j$-th stage of constructing $E''$, the execution
reaches configuration $C_j$, where the next step of $p_j$ is the
effective linearization step of
\textsc{Invoke}$(j,op,arg)$.

In execution $E'$, from configuration $C_j$ until configuration
$D_3'$, every step taken is either:
(i) a step of a \textsc{Read} operation,
(ii) a step of a \textsc{Click} operation, or
(iii) a step of an \textsc{Invoke}$(i,op_2,arg_2)$ for $i \neq j$.
As before, only $p_{read}$ performs \textsc{Read} steps,
and it performs no \textsc{Invoke} or \textsc{Click} steps. Additionally, $E'$ is a \execType execution. Thus, the conditions of
Lemma~\ref{lemma:invoke_independence_to_read_k_bound}
hold here as well.

Therefore, the next step of $p_j$ from configuration $D_3'$
is the effective linearization step of
\textsc{Invoke}$(j,op,arg)$.
By construction, when $p_j$ takes this step, the execution
reaches configuration $C'$.

Since no other updating operations are applied to
$\mathcal{O}[j]$, every \textsc{Read}$(j)$ operation that
completes before this step, in particular, $op_2'$,
which completes before configuration $D_3'$, returns the
value preceding the \textsc{Invoke}.
Hence, the linearization point of $op_2'$ must occur before
the linearization point of $op_1'$.
By transitivity, the linearization point of $op_1'$
occurs after the linearization point of $op'$.
\end{proof}

By Lemma~\ref{lemma:k_bounded_order_of_linearization}, the snapshot contains different states of $\mathcal{O}[j]$ in the two executions. A subsequent \textsc{Observe} operation must reveal this difference.

By Lemma~23, $C\isim{p_{click}}C'$ and the two configurations have identical shared-memory states. We may therefore repeat the argument from Section~\ref{sec:first_impossible}. Lemma~\ref{lemma:indistinguishable_configurations} implies that any sequence of events $\sigma$ consisting solely of steps by process $p_{click}$ that can occur starting at $C$ can also occur starting at $C'$. Moreover, if $\sigma$ is finite, then $C\sigma \isim{p_{click}} C'\sigma$.

Let $\sigma$ be the sequence of events in which $p_{click}$ executes \textsc{Observe}$(j)$ in isolation starting at configuration~$C$. Since $\mathcal{A}$ is obstruction-free (Property~\ref{prop:obstruction_free} of Definition~\ref{def:properties_non_intrusive}), the execution from $C$ terminates and, by Lemma~\ref{lemma:k_bounded_order_of_linearization}, returns $o_2$. By Lemmas~\ref{lemma:k_bounded_conf_similar} and~\ref{lemma:indistinguishable_configurations}, the same sequence can occur from $C'$. It must also terminate and, by Lemma~\ref{lemma:k_bounded_order_of_linearization}, return $o_1$. However, responses are part of the local state of $p_{click}$, and Lemma~\ref{lemma:indistinguishable_configurations} gives $C\sigma \isim{p_{click}}C'\sigma$. Thus, the operations cannot return different values, a contradiction.

This contradiction completes the proof that no linearizable adaptive snapshot with a $k$-bounded click algorithm can employ invisible \textsc{Read} operations. This concludes the proof of Theorem~\ref{th:main_non_intrusive}.

\section{Conclusion} \label{sec:conclusion}

In many practical workloads, read operations are significantly more frequent than update operations. It is therefore desirable for an adaptive snapshot algorithm to implement reads that are as efficient as possible, ideally invisible. However, despite this appeal, we show in this paper that, under natural assumptions satisfied by all existing adaptive snapshot implementations, such an implementation is impossible to design.
This result rules out a promising avenue for improving the performance of adaptive snapshots, which are among the most efficient snapshot constructions currently available. A natural question is how to overcome this barrier.

One possible direction is to design snapshot mechanisms tailored to higher-level data structures. In such settings, a “read” is typically a full data-structure operation rather than a simple memory access. Because these operations are already relatively expensive, the additional cooperation required to support snapshot scans may be less noticeable in practice.

Another direction is to design an adaptive snapshot that relaxes at least one of our assumptions and therefore may admit invisible reads. Constructing a practical algorithm of this kind remains an interesting open question.

\bibliography{ref_disc}

\appendix

\section{Examples of adaptive snapshots with oblivious click algorithms} \label{sec:examples_satisfy_props}

In this appendix, we consider two existing algorithms and show that they satisfy the definition of an adaptive snapshot with oblivious click algorithm (Definition~\ref{def:properties_oblivious_click}).

\subsection{Jayanti, Jayanti, and Jayanti~\cite{memsnap_podc_24}}

In the algorithm of Jayanti, Jayanti, and Jayanti~\cite{memsnap_podc_24}, the method used to modify components is denoted by \textsc{Update}$(i,op,arg)$. The state of the components is stored in an array named $A$, while the snapshot metadata is maintained in a variable named $X$ and an additional array named $B$.

Observe that \textsc{Click} operations only increment $X$. Furthermore, all operations are wait-free and hence obstruction-free. Hence, the algorithm satisfies Property~\ref{prop:oblivious_limit_click} and Property~\ref{prop:oblivious_obstruction_free}.

In an \textsc{Update}$(i,op,arg)$ operation, the component itself (i.e., $A[i]$) is modified only in Line~$7$, where $op(arg)$ is directly applied to $A[i]$. Therefore, the effective linearization step of the operation must occur at the unique line at which the component is modified, namely, Line~$7$. Any \textsc{Read}$(i)$ operation invoked after the execution of this line observes the updated value. Although such a \textsc{Read}$(i)$ operation invokes \textsc{Forward}$(i)$ before reading the new value, which can modify the snapshot metadata, it must nevertheless return the updated value read from $A[i]$. Hence, by definition, Line~$7$ constitutes the effective linearization step of the operation.

If a process halts before executing Line~7, any intervening \textsc{Click} steps can change only $X$. These changes do not affect the halted process's local state or its next step. When the process resumes, it executes Line~7, which remains the effective linearization step. Hence, the algorithm satisfies Property~\ref{prop:oblvious_click_linearize} and Property~\ref{prop:oblivious_act_on_memory}.

We therefore conclude that this algorithm is an adaptive snapshot with oblivious click algorithm.

\subsection{Bashari, Chan, and Woelfel~\cite{bashari_et_al:LIPIcs.DISC.2024.7}}

In the algorithm of Bashari, Chan, and Woelfel~\cite{bashari_et_al:LIPIcs.DISC.2024.7}, the state of the components is stored in an array denoted by $O$, while the snapshot metadata is maintained in a variable named $clk$ and in three arrays named $lastScan$, $lastUpdate$, and $R$.

Observe that \textsc{Click} operations only increment $clk$ and access or modify the $lastScan$ array. Moreover, all operations are wait-free and hence obstruction-free. Hence, the algorithm satisfies Property~\ref{prop:oblivious_limit_click} and Property~\ref{prop:oblivious_obstruction_free}.

In an \textsc{Invoke}$(i,op,arg)$ operation, the component itself (i.e., $O[i]$) is modified only in Line~$37$, where $op(arg)$ is directly applied to $O[i]$. Therefore, the effective linearization step of the operation must occur at the unique line at which the component is modified, namely, Line~$37$. Any \textsc{Read}$(i)$ operation invoked after the execution of this line observes the updated value. Although such a \textsc{Read}$(i)$ operation invokes \textsc{HelpUpdate}, which can modify the snapshot metadata, it must nevertheless return the updated value read from $O[i]$. Hence, by definition, Line~$37$ constitutes the effective linearization step of the operation.

If a process halts before executing Line~37, any intervening \textsc{Click} steps can change only $clk$ and $lastScan$. These changes do not affect the halted process's local state or its next step. When the process resumes, it executes Line~7, which remains the effective linearization step. Hence, the algorithm satisfies Property~\ref{prop:oblvious_click_linearize} and Property~\ref{prop:oblivious_act_on_memory}.

We therefore conclude that this algorithm is an adaptive snapshot with oblivious click algorithm.

\section{Detailed proofs} \label{sec:additional_proofs}

In this appendix, we present detailed proofs of results from earlier sections that were postponed due to space constraints. Before each proof, we restate the relevant claim.

\subsection{Proof of Observation~\ref{obs:updating_op_globabl}} \label{proof:updating_op_globabl}

Observation~\ref{obs:updating_op_globabl}. {\it Let $E$ be a linearizable \execType\ execution of an adaptive
snapshot algorithm $\mathcal{A}$, and let
$OP=\textsc{Invoke}(i,op,arg)$ be an operation in~$E$.
Suppose that $OP$ is an updating operation with respect to some
linearization of~$E$, and that, when the operations are executed
sequentially according to this linearization, $OP$ changes the
state of component~$i$ from \texttt{state$_1$} to
\texttt{state$_2$}. Then, in every linearization of~$E$,
$OP$ is an updating operation, and when the operations are executed
sequentially according to the linearization, $OP$ changes the state
of component~$i$ from \texttt{state$_1$} to
\texttt{state$_2$}. }

\begin{proof}
Since $E$ is a \execType\ execution (see
Definition~\ref{def:exec_type}), \textsc{Invoke} operations on the
same component are never concurrent. Hence, their relative
linearization order is uniquely determined by real-time order.

Furthermore, by Definition~\ref{def:adaptive_snapshot}, the state of
a component is determined solely by the \textsc{Invoke} operations
applied to that component. Therefore, the state of component~$i$
immediately before $OP$ in the sequential execution induced by any
linearization of~$E$ is the same. Since, by assumption, there exists
a linearization in which this state is \texttt{state$_1$} and applying
$OP$ changes it to \texttt{state$_2$}, the same pre-state and state
transition occur in every linearization of~$E$.

Consequently, in every linearization of~$E$, $OP$ is an updating
operation, and when the operations are executed sequentially
according to that linearization, $OP$ changes the state of
component~$i$ from \texttt{state$_1$} to \texttt{state$_2$}.
\end{proof}

\subsection{Proof of Observation~\ref{obs:elp_exist}} \label{proof:elp_exist}

Observation~\ref{obs:elp_exist}. {\it Let $\mathcal{A}$ be a linearizable adaptive snapshot algorithm. Assume that \textsc{Read} operations in $\mathcal{A}$ are obstruction-free. Let $E$ be a \execType execution of $\mathcal{A}$. Then, every completed updating operation in $E$ has a unique effective linearization step. }

\begin{proof}
Let $OP=\textsc{Invoke}(i,op,arg)$ be a completed updating operation in $E$ that changes the state of component~$i$ from \texttt{state$_1$} to \texttt{state$_2$}. Since in $E$ there is at most one concurrent \textsc{Invoke} operation to each component, the state of component~$i$ at the invocation of $OP$ must be \texttt{state$_1$}, and its state upon termination of $OP$ must be \texttt{state$_2$}.

Since \textsc{Read} operations in $\mathcal{A}$ are
obstruction-free, any \textsc{Read} operation that executes in
isolation must terminate and return a value. Consequently, a \textsc{Read}$(i)$ operation that is invoked and executed in isolation before the invocation of $OP$ must return \texttt{state$_1$}, whereas a \textsc{Read}$(i)$ operation invoked and executed in isolation after the termination of $OP$ must return \texttt{state$_2$}. Therefore, there exists a first step between the invocation and the termination of $OP$ after which an isolated \textsc{Read}$(i)$ operation returns \texttt{state$_2$}. By definition, this step is the effective linearization step of $OP$.

Uniqueness follows from the minimality of this step.
\end{proof}

\subsection{Proof of Lemma~\ref{lemma:invoke_independence_to_read}} \label{proof:invoke_independence_to_read}

Lemma~\ref{lemma:invoke_independence_to_read}. {\it Let $\mathcal{A}$ be an adaptive snapshot algorithm that satisfies update linearization independence from \textsc{Click}, and implements invisible \textsc{Read} operations. Let $E$ be a \execType execution of $\mathcal{A}$ with history
$e_1, e_2, \ldots$ (finite or infinite). Let $p$ be a process that executes an effective linearization step of an updating operation \textsc{Invoke}$(i,op,arg)$.
 Let $e_t$ denote the step that is the effective linearization step of this operation in $E$. Consider an execution $E'$ whose first $t-1$ events are identical to the first
$t-1$ events of $E$, followed by an arbitrary sequence of steps, not executed by $p$, consisting solely of \textsc{Click} and \textsc{Read} operations (either continuing
previously invoked operations or newly invoked ones), such that no process executes steps from both \textsc{Click} and \textsc{Read} operations in that sequence. After these steps, process $p$ performs step $e_t$, followed by any valid suffix of execution
steps. Then, in execution $E'$, the effective linearization step of
\textsc{Invoke}$(i,op,arg)$ is the event $e_t$, performed by $p$, and the
result of the operation (both the state modification and the returned value)
is identical to its result in execution $E$. }

\begin{proof}
Let $E,i,op,arg,e_t$, and $p$ be as in the lemma.
Assume, for the sake of contradiction, that the lemma does not hold.
Then, there exists an execution $\tilde{E}$ whose first $t-1$ events are
identical to those of $E$, followed by a sequence of events $\sigma$
starting at configuration $D$, consisting solely of steps of
\textsc{Click} or \textsc{Read} operations, such that no process executes
steps from both \textsc{Click} and \textsc{Read} operations within that
sequence, and after $\sigma$, process $p$ performs step $e_t$, yet this
step is \emph{not} the effective linearization step of
\textsc{Invoke}$(i,op,arg)$.

Since \textsc{Read} operations are invisible (i.e., they do not modify
the shared memory), they affect only the local state of the invoking
process. Let $\sigma'$ be the sequence obtained from $\sigma$ by
removing all steps that belong to \textsc{Read} operations while
preserving the relative order of the remaining events.

We now prove that $\sigma'$ can occur from $D$. Let $n$ be the number of
steps in $\sigma$ that belong to \textsc{Click} operations. Then
$\sigma'$ can be written as $\sigma' = s_1, \ldots, s_n$. Let $P$ be the
set of processes that execute these steps, and for every
$1 \leq j \leq n$, let $\sigma_j'$ denote the prefix of $\sigma'$ up to
the $j$-th step, i.e., $\sigma_j' = s_1, \ldots, s_j$. Additionally, for
$1 \leq j \leq n$, let $D_j$ be the configuration obtained after
executing the $j$-th \textsc{Click} step of $\sigma$ from $D$. We prove the following subclaim by induction.

\paragraph*{Subclaim for Proof~\ref{proof:invoke_independence_to_read}}
For all $1 \leq j \leq n$, the sequence $\sigma_j'$ can occur from $D$.
Moreover, $D_j \isim{P} D\sigma_j'$, and the state of the shared memory in
these configurations is identical.

Observe that if the sub-claim holds for $j = n$, then $\sigma'$
can occur from $D$, $D_\sigma \isim{P} D_{\sigma'}$, and the shared
memory states in these configurations are identical.

\begin{proof}
For $j=1$, $\sigma_1'$ consists of a single step. Let $D'$ be the
configuration immediately preceding this step in the execution of $\sigma$ from $D$. Every prior step in $\sigma$ must belong to a
\textsc{Read} operation and therefore does not modify the shared
memory. Furthermore, no process in $P$ has taken any prior step in
$\sigma$, because no process executes steps from both
\textsc{Click} and \textsc{Read} operations in $\sigma$. Hence, the local state of
every process in $P$ is identical in $D$ and $D'$ (i.e.,
$D \isim{P} D'$), and the shared-memory states of $D$ and $D'$ coincide.
By Lemma~\ref{lemma:indistinguishable_configurations}, it follows that
$\sigma_1'$ can also occur from $D$, that
$D'\sigma_1' \isim{P} D\sigma_1'$, and that the shared-memory states
remain identical. Since $D'\sigma_1' = D_1$, the base case holds.

Assume the claim holds for some $j$, and consider $j+1$. By the
induction hypothesis, $D\sigma_j' \isim{P} D_j$, and the shared-memory
states of these configurations are identical. Let $D'$
be the configuration immediately preceding step $j+1$ in the execution of $\sigma$ from $D$. Every step between $D_j$ and $D'$ in the execution of $\sigma$ starting at $D$ 
belongs to a \textsc{Read} operation and thus does not modify the shared
memory. Moreover, no process in $P$ takes steps between these
configurations, since no process executes both \textsc{Click} and
\textsc{Read} steps in $\sigma$. Hence,
$D_j \isim{P} D'$, and the shared-memory states of $D_j$ and $D'$ are
identical. Together with the induction hypothesis and
Lemma~\ref{lemma:indistinguishable_configurations}, it follows that
$s_{j+1}$ can occur from $D\sigma_j'$, that
$D's_{j+1} \isim{P} D\sigma_j's_{j+1}$, and that the shared-memory
states remain identical. Since $D's_{j+1} = D_{j+1}$ and
$\sigma_j's_{j+1} = \sigma_{j+1}'$, the induction step holds.
\end{proof}

We now return to the main proof. Because $A$ satisfies update linearization independence from \textsc{Click} (Property~\ref{prop:oblvious_click_linearize} of
Definition~\ref{def:properties_oblivious_click}) and $\sigma'$ contains only \textsc{Click} steps, if $\sigma'$ is executed from $D$ and $p$ then performs $e_t$, that step is the effective linearization step of the \textsc{Invoke} operation. 

Let $q$ be an idle process that does not participate in either $E$ or $\tilde{E}$. We first rule out the possibility that the effective
linearization step occurs prior to $e_t$, namely, that some step in
$\sigma$ causes subsequent \textsc{Read} operations to observe the
updated state. Suppose this happens, and let $e \in \sigma$ be the first
such step when $\sigma$ is executed from $D$. Since \textsc{Read}
operations do not modify shared memory, $e$ must be part of a
\textsc{Click} operation and thus also appears in $\sigma'$. Hence,
$e = s_j$ for some $1 \leq j \leq n$. The configuration after $e$ in $\sigma$ is $D_i$, and the configuration after $e$ in $\sigma'$ is $D\sigma_j'$. By the subclaim, the shared-memory states of $D_j$ and $D\sigma_j'$ are identical. Since $q$ does not take steps in $\sigma$, we also have $D_j \isim{q} D\sigma_
j'$. By Lemma~\ref{lemma:indistinguishable_configurations}, if $q$ now performs $\textsc{Read}(i)$ from either configuration, it must return the updated value in both executions. The \textsc{Read}$(i)$ operation is performed after $\sigma'_j$ but before $e_t$. Since $\sigma'_j$ contains only \textsc{Click} steps, Property~\ref{prop:oblvious_click_linearize} of Definition~\ref{def:properties_oblivious_click} implies that the effective linearization step remains $e_t$. Hence, the \textsc{Read} should return the previous value, a contradiction.

Finally, we have that the shared-memory states of $D\sigma$ and
$D\sigma'$ are identical. Since neither $p$ nor $q$ take steps in
$\sigma$ or $\sigma'$, we also have
$D\sigma \isim{\{p,q\}} D\sigma'$. By
Lemma~\ref{lemma:indistinguishable_configurations}, if $p$ performs
$e_t$ and then $q$ executes $\textsc{Read}(i)$ from either configuration,
the returned values must coincide. Since executing $e_t$ from
$D\sigma'$ causes subsequent reads to observe the updated state, the
same must hold from $D\sigma$. Therefore, in the execution $\tilde{E}$,
$e_t$ is the first step after which \textsc{Read} operations
observe the updated value, implying that $e_t$ is the effective
linearization step of the operation, a contradiction. Hence, the lemma
holds.
\end{proof}

\subsection{Proof of Lemma~\ref{lemma:execution_valid}} \label{proof:execution_valid}

Lemma~\ref{lemma:execution_valid}. {\it Execution $E''$ in Subsection~\ref{subsec:first_construction} is a valid finite execution of algorithm~$\mathcal{A}$.}

\begin{proof}

    Let the execution start at an initial configuration $D_0$. For every $i\geq 1$, let $C_i$ be the configuration after Stage~\ref{step:oblivious_read} in the $i$-th iteration of execution $E''$, and let $D_i$ be the configuration after Stage~\ref{step:oblivious_next_event} in the $i$-th iteration. By obstruction-freedom (Property~\ref{prop:oblivious_obstruction_free} of Definition~\ref{def:properties_oblivious_click}),
    each execution of Stage~\ref{step:oblivious_read} terminates after a finite
    number of steps. Therefore, every iteration terminates after a finite number of steps. Since \textsc{Read} operations are invisible, executing
    Stage~\ref{step:oblivious_read} does not modify the shared memory.
    Additionally, process $p_1$ takes no steps in
    Stage~\ref{step:oblivious_read}. Hence, for every $i\ge 1$, $C_i \isim{p_1} D_{i-1}$, and the shared-memory state is identical in configurations $C_i$ and $D_{i-1}$.

    Consider an execution $E'''$ that starts from configuration $D_0$ and in
    which only $p_1$ takes steps. By obstruction-freedom (Property~\ref{prop:oblivious_obstruction_free} of Definition~\ref{def:properties_oblivious_click}), the \textsc{Invoke}
    operation executed by $p_1$ terminates after a finite number of steps. Let $F_i$ be
    the configuration obtained after $p_1$ performs its first $i$ steps from
    $D_0$ in $E'''$. We claim that for every $i\ge 1$, $F_i \isim{p_1} D_i$, and the shared-memory state is identical in $F_i$ and $D_i$.

    We prove the claim by induction on $i$. For the base case, $D_0\isim{p_1}C_1$, and the two configurations have identical shared-memory states. Lemma~\ref{lemma:indistinguishable_configurations} therefore implies that executing the first step of $p_1$ from these configurations yields $F_1\isim{p_1}D_1$ with identical shared-memory states.

    For the induction step, assume the claim holds for $i-1$, i.e.,
    $F_{i-1} \isim{p_1} D_{i-1}$ and the two configurations have identical shared-memory states. Since $D_{i-1} \isim{p_1} C_i$ and the shared-memory state is identical in $D_{i-1}$ and $C_i$, we obtain that $F_{i-1} \isim{p_1} C_i$, with identical shared-memory state as well. Lemma~\ref{lemma:indistinguishable_configurations} therefore implies that executing the next step of $p_1$ from these configurations yields $F_i\isim{p_1}D_i$ with identical shared-memory states.

    Since the \textsc{Invoke} operation terminates in $E'''$ and it modifies the state of $\mathcal{O}[1]$, it has an
    effective linearization step, which must be executed by $p_1$ (since $p_{read}$ does not modify shared memory). Let $F_j$ be the
    configuration immediately after this event. Because no other \textsc{Invoke}
    operations are applied to~$\mathcal{O}[1]$, every \textsc{Read}$(1)$ invoked
    after $F_j$ and executed in isolation returns the new state, namely $o_2$.
    Because the shared-memory states of $F_j$ and $D_j$ are identical, a \textsc{Read}$(1)$
    invoked after $D_j$ and executed by $p_{read}$ in isolation must also return $o_2$.

    Consequently, the loop terminates after exactly $j+1$ iterations. Suppose that the
    \textsc{Invoke} operation terminates at some earlier iteration. By
    linearizability, in the subsequent iteration $p_{read}$ must read the updated state.
    Hence, the earliest iteration in which the \textsc{Invoke} operation can
    terminate is the $j$-th iteration. It follows that the execution $E''$ is
    well-defined and finite.
\end{proof}

\subsection{Proof of Lemma~\ref{lemma:invoke_independence_to_read_k_bound}} \label{proof:invoke_independence_to_read_k_bound}

Lemma~\ref{lemma:invoke_independence_to_read_k_bound}. {\it Let $\mathcal{A}$ be a linearizable adaptive snapshot with a $k$-bounded click algorithm that implements invisible \textsc{Read} operations, and $E$ be a \execType execution of $\mathcal{A}$ with history
$e_1, e_2, \ldots$ (finite or infinite). Let $p$ be a process that executes the effective linearization step of an updating operation \textsc{Invoke}$(i,op,arg)$. Let $e_t$ denote the step that is the effective linearization step of this operation in $E$. Consider a \execType execution $E'$, whose first $t-1$ events are identical to the first $t-1$ events of $E$, followed by an arbitrary sequence of steps, none of which are executed by $p$, consisting solely of \textsc{Click} and \textsc{Read} operations, or \textsc{Invoke}$(j,op_2,arg_2)$ operations with $j \neq i$ (either continuing previously invoked operations or newly invoked ones), such that no process executes steps of both \textsc{Read} and \textsc{Click}, or both \textsc{Read} and \textsc{Invoke}, within that sequence. After these steps, process $p$ performs step $e_t$, followed by any valid suffix of execution steps. Then, in execution $E'$, the effective linearization step of \textsc{Invoke}$(i,op,arg)$ is the event $e_t$ performed by $p$, and the result of the operation (both the state modification and the returned value) is identical to its result in execution $E$.
}

\begin{proof}
Let $E,i,op,arg,e_t$, and $p$ be as in the lemma.
Assume, for the sake of contradiction, that the lemma does not hold.
Then, there exists an execution $\tilde{E}$ whose first $t-1$ events are
identical to those of $E$, followed by a sequence of events $\sigma$
starting at configuration $D$, consisting solely of \textsc{Click} and \textsc{Read} operations, or \textsc{Invoke}$(j,op_2,arg_2)$ operations with $r \neq i$, such that no process executes steps from both \textsc{Read} and \textsc{Click}, or from both \textsc{Read} and \textsc{Invoke}, within that sequence, and after $\sigma$, process $p$ performs step $e_t$, yet this
step is \emph{not} the effective linearization step of
\textsc{Invoke}$(i,op,arg)$.

Since \textsc{Read} operations are invisible (i.e., they do not modify
the shared memory), they affect only the local state of the invoking
process. Let $\sigma'$ be the sequence obtained from $\sigma$ by
removing all steps that belong to \textsc{Read} operations while
preserving the relative order of the remaining events.

We now prove that $\sigma'$ can occur from $D$. Let $n$ be the number of
steps in $\sigma$ that belong to \textsc{Click} or \textsc{Invoke} operations. Then, $\sigma'$ can be written as $\sigma' = s_1, \ldots, s_n$. Let $P$ be the
set of processes that execute these steps, and for every
$1 \leq j \leq n$, let $\sigma_j'$ denote the prefix of $\sigma'$ up to the $j$-th step, i.e., $\sigma_j' = s_1, \ldots, s_j$. Additionally, for
$1 \leq j \leq n$, let $D_j$ be the configuration obtained after
executing the $j$-th \textsc{Click} or \textsc{Invoke} step of $\sigma$ from $D$. We prove the following subclaim by induction

\paragraph*{Subclaim for Proof~\ref{proof:invoke_independence_to_read_k_bound}}
For all $1 \leq j \leq n$, the sequence $\sigma_j'$ can occur from $D$.
Moreover, $D_j \isim{P} D\sigma_j'$, and the state of the shared memory in
these configurations is identical.

Observe that if the sub-claim holds for $j = n$, then $\sigma'$
can occur from $D$, $D_\sigma \isim{P} D_{\sigma'}$, and the shared
memory states in these configurations are identical.

\begin{proof}
For $j=1$, $\sigma_1'$ consists of a single step. Let $D'$ be the
configuration immediately preceding this step in the execution of $\sigma$ from $D$. Every prior step in $\sigma$ must belong to a
\textsc{Read} operation and therefore does not modify the shared
memory. Furthermore, no process in $P$ has taken any prior step in
$\sigma$, because no process executes steps from both \textsc{Read} and \textsc{Click}, or from both \textsc{Read} and \textsc{Invoke} in $\sigma$. Hence, the local state of
every process in $P$ is identical in $D$ and $D'$ (i.e.,
$D \isim{P} D'$), and the shared-memory states of $D$ and $D'$ coincide.
By Lemma~\ref{lemma:indistinguishable_configurations}, it follows that
$\sigma_1'$ can also occur from $D$, that
$D'\sigma_1' \isim{P} D\sigma_1'$, and that the shared-memory states
remain identical. Since $D'\sigma_1' = D_1$, the base case holds.

Assume the claim holds for some $j$, and consider $j+1$. By the
induction hypothesis, $D\sigma_j' \isim{P} D_j$, and the shared-memory
states of these configurations are identical. Let $D'$ be the configuration immediately preceding step $j+1$ in the execution of $\sigma$ from $D$. Every step between $D_j$ and $D'$ in the execution of $\sigma$ starting at $D$ 
belongs to a \textsc{Read} operation and thus does not modify the shared
memory. Moreover, no process in $P$ takes steps between these
configurations, since no process executes steps of both \textsc{Read} and \textsc{Click} operations or of both \textsc{Read} and \textsc{Invoke} operations in $\sigma$. Hence,
$D_j \isim{P} D'$, and the shared-memory states of $D_j$ and $D'$ are
identical. Together with the induction hypothesis and
Lemma~\ref{lemma:indistinguishable_configurations}, it follows that
$s_{j+1}$ can occur from $D\sigma_j'$, that
$D's_{j+1} \isim{P} D\sigma_j's_{j+1}$, and that the shared-memory
states remain identical. Since $D's_{j+1} = D_{j+1}$ and
$\sigma_j's_{j+1} = \sigma_{j+1}'$, the induction step holds.
\end{proof}

We now return to the main proof. Because $\mathcal{A}$ satisfies update linearization independence from \textsc{Click} and \textsc{Invoke}, and because $\sigma'$ contains only such steps and remains component-serialized, executing $\sigma'$ from $D$ and then having $p$ perform $e_t$ leaves $e_t$ as the effective linearization step of the \textsc{Invoke} operation.

Let $q$ be an idle process that does not participate in either $E$ or $\tilde{E}$. We first rule out the possibility that the effective linearization step occurs prior to $e_t$, namely, that some step in $\sigma$ causes subsequent \textsc{Read} operations to observe the updated state. Suppose this happens, and let $e \in \sigma$ be the first such step when $\sigma$ is executed from $D$. Since \textsc{Read} operations do not modify shared memory, $e$ must be part of a \textsc{Click} or \textsc{Invoke}$(r,op_2,arg_2)$ with $r \neq i$, and thus also appears in $\sigma'$. Hence, $e = s_j$ for some $1 \leq j \leq n$. The configuration after $e$ in $\sigma$ is $D_j$, and the configuration after $e$ in $\sigma'$ is $D\sigma_j'$. By the subclaim, the shared-memory states of $D_j$ and $D\sigma_j'$ are identical. Since $q$ does not take steps in $\sigma$, we also have $D_j \isim{q} D\sigma_j'$. By Lemma~\ref{lemma:indistinguishable_configurations}, if $q$ now performs $\textsc{Read}(i)$ from either configuration, it must return the updated value in both executions.

We observe that the \textsc{Read}$(i)$ operation is performed after executing $\sigma_j'$ but before executing $e_t$. Furthermore, $\sigma_j'$ consists of steps only from \textsc{Click}, or \textsc{Invoke}$(r,op_2,arg_2)$ with $r \neq i$, and the execution obtained when executing $\sigma_j'$ starting at $D$ is a \execType execution. Therefore, Property~\ref{prop:non_intrusive_linearize} of Definition~\ref{def:properties_non_intrusive} implies that in the obtained execution the effective linearization step remains $e_t$, and thus the \textsc{Read}$(i)$ operation should return the previous value, a contradiction.

Finally, we have that the shared-memory states of $D\sigma$ and
$D\sigma'$ are identical. Since neither $p$ nor $q$ take steps in
$\sigma$ or $\sigma'$, we also have
$D\sigma \isim{\{p,q\}} D\sigma'$. By
Lemma~\ref{lemma:indistinguishable_configurations}, if $p$ performs
$e_t$ and then $q$ executes $\textsc{Read}(i)$ from either configuration,
the returned values must coincide. Since executing $e_t$ from
$D\sigma'$ causes subsequent reads to observe the updated state, the
same must hold from $D\sigma$. Therefore, in the execution $\tilde{E}$,
$e_t$ is the first step after which \textsc{Read} operations
observe the updated value, implying that $e_t$ is the effective
linearization step of the operation, a contradiction. Hence, the lemma
holds.
\end{proof}

\subsection{Proof of Lemma~\ref{lemma:k_bounded_conf_similar}} \label{proof:k_bounded_conf_similar}

Lemma~\ref{lemma:k_bounded_conf_similar}. {\it Let $C$ and $C'$ be the configurations defined in Subsection~\ref{subsec:second_construction}. Then
$C \isim{p_{click}} C'$, and the shared-memory state in both configurations
is identical.}

\begin{proof}
We begin by comparing configurations $Q_{j-1}$ and $Q_j$.
Executions $E$ and $E'$ are identical up to configuration $Q_{j-1}$.
In execution $E$, the transition from $Q_{j-1}$ to $Q_j$
consists of a single step taken by $p_j$.

Recall that during the construction of $E''$ there is at most one concurrent \textsc{Invoke} operation per component. Additionally, the $j$-th stage of the construction ends in configuration $C_j$ such that the next step of $p_j$
is the effective linearization step of
\textsc{Invoke}$(j,op,arg)$.
In execution $E$, from configuration $C_j$ until $Q_{j-1}$,
every step taken is either
(i) a step of a \textsc{Read} operation, or
(ii) a step of an \textsc{Invoke}$(i,op_2,arg_2)$ operation for
$i \neq j$. Moreover, only $p_{read}$ performs \textsc{Read} steps, and it performs no \textsc{Invoke} or \textsc{Click} steps. The execution $E$ is component-serialized. Hence, the conditions of Lemma~\ref{lemma:invoke_independence_to_read_k_bound} hold. It follows that the next step of $p_j$ from $Q_{j-1}$
is the effective linearization step of
\textsc{Invoke}$(j,op,arg)$.

By Property~\ref{prop:act_on_memory} of
Definition~\ref{def:properties_non_intrusive},
this step modifies no shared-memory locations
except possibly shared memory associated with $\mathcal{O}[j]$.
Since $p_{click}$ has not taken any steps so far,
its local state is identical in $Q_{j-1}$ and $Q_j$,
and the shared memory differs only at shared-memory locations associated with $\mathcal{O}[j]$.

\medskip
Next, in both executions, processes
$p_{j+1},\ldots,p_{k+1}$ each take one step.
By the same argument as above, each such step
is the effective linearization step of
\textsc{Invoke}$(i,op,arg)$ for the corresponding $i$.
By Property~\ref{prop:act_on_memory} of Definition~\ref{def:properties_non_intrusive},
none of these steps access $\mathcal{O}[j]$.
Thus, they do not access the only shared-memory locations
at which $Q_{j-1}$ and $Q_j$ may differ.
Applying Lemma~\ref{lemma:indistinguishable_configurations},
the resulting configurations $Q_{k+1}$ and $Q_{k+1}'$
differ only in shared-memory locations associated with~$\mathcal{O}[j]$, and since $p_{click}$ has still taken no steps,
we have $Q_{k+1} \isim{p_{click}} Q_{k+1}'$.

Configurations $D_1$ and $D_1'$ are obtained by having
$p_{read}$ complete \textsc{Read}$(i)$
for every $1 \le i \le k+1$,
starting at $Q_{k+1}$ and $Q_{k+1}'$, respectively.
Since reads are invisible in $\mathcal{A}$,
they do not modify shared memory,
and $p_{click}$ still takes no steps.
Hence, $D_1 \isim{p_{click}} D_1'$, and the shared memory may differ only in shared-memory locations associated with~$\mathcal{O}[j]$.

Next, configurations $C$ and $D_2'$ are obtained by
having $p_{click}$ complete a \textsc{Click}
operation starting at $D_1$ and $D_1'$, respectively.
By the choice of $j$ and Property~\ref{prop:limit_memory} of Definition~\ref{def:properties_non_intrusive},
this \textsc{Click} operation in $E$ does not access
$\mathcal{O}[j]$.
Since $D_1$ and $D'_1$ differ only in shared-memory locations associated with~$\mathcal{O}[j]$, the sets of locations accessed by the \textsc{Click} operation in $E$ and $E'$
are identical.
By Lemma~\ref{lemma:indistinguishable_configurations},
executing the \textsc{Click} operation from $D_1$ and $D_1'$
produces identical modifications to shared memory
and leaves $p_{click}$ in the same local state.
Therefore, $C \isim{p_{click}} D_2'$, and the shared memory differ only in shared-memory locations associated with~$\mathcal{O}[j]$.

From $D_2'$, configuration $D_3'$ is obtained by
$p_{read}$ completing another sequence of
\textsc{Read}$(i)$ operations.
Again, since reads are invisible,
shared memory remains unchanged and
$p_{click}$’s local state is unaffected.
Thus, $C \isim{p_{click}} D_3'$, and the shared memory differ only in shared-memory locations associated with~$\mathcal{O}[j]$

Finally, configuration $C'$ is obtained from $D_3'$
by a single step of $p_j$.
Let $e$ denote the step that moves execution $E$
from $Q_{j-1}$ to $Q_j$. The prefix of $E'$ from $Q_{j-1}$ to $D'_3$ contains only \textsc{Read}, \textsc{Click}, and \textsc{Invoke}$(i,op_2,arg_2)$ steps with $i\neq j$. Only $p_{read}$ performs \textsc{Read} steps, and it performs no \textsc{Click} or \textsc{Invoke} steps. Both the prefix of $E$ up to $Q_{j-1}$ and execution $E'$ are component-serialized. Thus, the conditions of Lemma~\ref{lemma:invoke_independence_to_read_k_bound}
hold once more. It follows that the next step of $p_j$
from $D_3'$ is the effective linearization step of \textsc{Invoke}$(j,op,arg)$, and it modifies $\mathcal{O}[j]$
exactly as step $e$ does in $E$.

The only difference between the shared-memory states of $C$ and $D'_3$ is the effect of step $e$ on the shared-memory locations associated with~$\mathcal{O}[j]$. Executing this step from $D'_3$ therefore yields a configuration $C'$ with the same shared-memory state as $C$. Moreover, $p_{click}$ takes no steps between $D_3'$
and $C'$. Because $D_3' \isim{p_{click}} C$,
we conclude that $C' \isim{p_{click}} C$, and the shared-memory states are identical.
\end{proof}

\end{document}